\documentclass{amsart}%
\usepackage{amsmath}
\usepackage{amssymb}
\usepackage{amsfonts}
\usepackage{graphicx}%
\newtheorem{theorem}{Theorem}
\theoremstyle{plain}
\newtheorem{acknowledgement}{Acknowledgement}

\newtheorem{corollary}{Corollary}

\newtheorem{definition}{Definition}

\newtheorem{lemma}{Lemma}
\newtheorem{notation}{Notation}

\newtheorem{proposition}{Proposition}
\newtheorem{remark}{Remark}

\numberwithin{equation}{section}
\begin{document}
\title[Quantum \ Field Formulation of p-Adic string amplitudes ]{Euclidean Quantum \ Field Formulation of p-Adic open string amplitudes }
\author[Fuquen-Tibat\'{a}]{A. R. Fuquen-Tibat\'{a}}
\address{Centro de Investigaci\'{o}n y de Estudios Avanzados del Instituto
Polit\'{e}cnico Nacional\\
Departamento de Matem\'{a}ticas, Unidad Quer\'{e}taro\\
Libramiento Norponiente \#2000, Fracc. Real de Juriquilla. Santiago de
Quer\'{e}taro, Qro. 76230\\
M\'{e}xico.}
\email{arfuquen@math.cinvestav.mx}
\author[Garc\'{\i}a-Compe\'{a}n]{H. Garc\'{\i}a-Compe\'{a}n}
\address{Centro de Investigacion y de Estudios Avanzados del I.P.N., Departamento de
F\'{\i}sica, Av. Instituto Politecnico Nacional 2508, Col. San Pedro
Zacatenco, Mexico D.F., C.P. 07360, Mexico}
\email{compean@fis.cinvestav.mx}
\author[Z\'{u}\~{n}iga-Galindo]{W. A. Z\'{u}\~{n}iga-Galindo}
\address{University of Texas Rio Grande Valley\\
School of Mathematical \& Statistical Sciences\\
One West University Blvd\\
Brownsville, TX 78520, United States and Centro de Investigaci\'{o}n y de
Estudios Avanzados del Instituto Polit\'{e}cnico Nacional\\
Departamento de Matem\'{a}ticas, Unidad Quer\'{e}taro\\
Libramiento Norponiente \#2000, Fracc. Real de Juriquilla. Santiago de
Quer\'{e}taro, Qro. 76230\\
M\'{e}xico.}
\email{wilson.zunigagalindo@utrgv.edu, wazuniga@math.cinvestav.edu.mx}

\begin{abstract}
We study in a rigorous mathematical way $p$-adic quantum field theories whose
$N$-point amplitudes are the expectation of products of vertex operators. We
show that this type of amplitudes admit a series expansion where each term is
an Igusa's local zeta function. The lowest term in this series is a
regularized version of the $p$-adic open Koba-Nielsen string amplitude.

\end{abstract}
\maketitle
\tableofcontents

\section{Introduction}

The string amplitudes were introduced by Veneziano in the 60s,
\cite{Veneziano}, further generalizations were obtained by Virasoro
\cite{Virasoro}, Koba and Nielsen \cite{Koba-Nielsen}, among others. In the
80s, Freund, Witten and Volovich, among others, studied string amplitudes at
the tree level over different number fields, and suggested the existence of
connections between these amplitudes, see e.g. \cite{B-F-O-W}-\cite{Vol}. In
this framework the connections with number theory, specifically with local
zeta functions appears naturally, see e.g. \cite{Bocardo:2020mk}%
-\cite{Zun-B-C-JHEP}, and the survey \cite{Symmetry}, see also
\cite{Arefeva-1}-\cite{Arefeva-3}.

The $p$-adic string theories have been studied over time with some periodic
fluctuations in their interest (for some reviews, see \cite{Brekke et al},
\cite{Hloseuk-Spect}, \cite{V-V-Z}, \cite{Dragovich:2017kge}). Recently a
considerable amount of work has been performed on this topic in the context of
the AdS/CFT correspondence
\cite{Gubser:2016guj,Heydeman:2016ldy,Gubser:2016htz,Dutta:2017bja}. String
theory with a $p$-adic world-sheet was proposed and studied for the first time
in \cite{Freund:1987kt}. Later this theory was formally known as $p$-adic
string theory. The $p$-adic strings are related to ordinary strings at least
in two different ways. First, connections through the adelic relations
\cite{Freund:1987ck}, and second, through the limit $p$ tends to $1$
\cite{Gerasimov:2000zp}-\cite{Ghoshal:2006}.

The tree-level string amplitudes were explicitly computed in the case of
$p$-adic string world-sheet in \cite{Brekke:1988dg} and \cite{Frampton:1987sp}%
. Since the 80s there has been interest in constructing field theories whose
correlators are the $p$-adic tree-level string amplitudes (or $p$-adic
Koba-Nielsen amplitudes). Spokoiny \cite{Spokoiny:1988zk} and Zhang
\cite{Zhang}, see also \cite{Parisi}, constructed formally quantum field
theories whose amplitudes are expectation values of products of vertex
operators. In \cite{Zabrodin:1988ep} Zabrodin established that the tree-level
string amplitudes may be obtained starting with a discrete field theory on a
Bruhat-Tits tree. These ideas have been used by Ghoshal and Kawano in the
study of $p$-adic strings in constant B-fields \cite{Ghoshal:2004ay}. This
article aims to provide a rigorous mathematical construction of a class of
quantum field theories whose amplitudes are expectations of products of vertex
operators. By using this approach, we carry out a mathematically rigorous
derivation of the $N$-point Koba-Nielsen amplitudes, thus our approach is
completely different from the one followed in
\cite{Spokoiny:1988zk,Zhang,Zabrodin:1988ep}.

The naive Euclidean version of the $p$-adic $N$-point amplitudes is given by
\begin{align}
\mathcal{A}^{\left(  N\right)  }\left(  \boldsymbol{k}\right)   &
=\left\langle
{\displaystyle\prod\limits_{j=1}^{N}}
\text{ }%
{\displaystyle\int\limits_{\mathbb{Q}_{p}}}
dx_{j}\text{{} }e^{\boldsymbol{k}_{j}\cdot\boldsymbol{\varphi}\left(
x_{j}\right)  }\right\rangle \label{Naive amplitude}\\
&  =\frac{1}{Z_{0}^{phys}}\int D\boldsymbol{\varphi}\text{{}}e^{-S\left(
\boldsymbol{\varphi}\right)  }\left\{
{\displaystyle\int\limits_{\mathbb{Q}_{p}^{N}}}
d^{N}x\text{{} }e^{\sum_{j=1}^{N}\boldsymbol{k}_{j}\cdot\boldsymbol{\varphi
}\left(  x_{j}\right)  }\right\}  ,\nonumber
\end{align}
where $\int_{\mathbb{Q}_{p}}dx_{j}${} $e^{\boldsymbol{k}_{j}\cdot
\boldsymbol{\varphi}\left(  x_{j}\right)  }$ is the tachyonic vertex operator
of the $j-$th tachyon, with momentum $\boldsymbol{k}_{j}=\left(
k_{0,j},\ldots,k_{D-1,j}\right)  $, and field $\boldsymbol{\varphi}(x_{j})$,
the dot denotes the standard Euclidean scalar product, and the action is given
by%
\begin{equation}
S\left(  \boldsymbol{\varphi}\right)  =\frac{T_{0}}{2}%
{\displaystyle\sum\limits_{j=1}^{N}}
\text{ }%
{\displaystyle\int\limits_{\mathbb{Q}_{p}}}
\text{ }%
{\displaystyle\int\limits_{\mathbb{Q}_{p}}}
\genfrac{\{}{\}}{}{}{\varphi_{j}\left(  x_{j}\right)  -\varphi_{j}\left(
y_{j}\right)  }{\left\vert x_{j}-y_{j}\right\vert _{p}}%
^{2}dx_{j}dy_{j}\text{.} \label{Action}%
\end{equation}
It is important to note that in (\ref{Naive amplitude}) the tachyonic fields
must be functions not distributions. These amplitudes are exactly the ones
considered in \cite{Spokoiny:1988zk}, \cite{Zhang}, \cite{Ghoshal:2004ay}.
Since the integral $\int_{\mathbb{Q}_{p}^{N}}d^{N}x$ in the right-hand side of
(\ref{Naive amplitude}) is always divergent, it is necessary to introduce a
cut-off, and to define the amplitude by a limit process. The key observation
is that the action (\ref{Action}) \ corresponds to a free quantum field. In
the Archimedean and non-Archimedean cases, free quantum fields correspond to
Gaussian probability measures on suitable infinite dimensional spaces. The
reader may consult \cite[Section 6.2]{Jaffe-Glimm} for the Archimedean case,
and \cite[Section 5.5]{Arroyo-Zuniga}, \cite{Zuniga-JFAA},
\cite{Zuniga-Preprint} for the $p$-adic case. We construct Gaussian
probability measure $\mathbb{P}_{D}$ on suitable function space
\ ($\mathcal{L}_{\mathbb{R}}^{D}\left(  \mathbb{Q}_{p}\right)  $) \ and
propose that $\mathcal{A}^{\left(  N\right)  }\left(  \boldsymbol{k}\right)
=\lim_{R\rightarrow\infty}\mathcal{A}_{R}^{\left(  N\right)  }\left(
\boldsymbol{k}\right)  $, where
\begin{equation}
\mathcal{A}_{R}^{\left(  N\right)  }\left(  \boldsymbol{k}\right)  =\frac
{1}{Z_{0}}%
{\displaystyle\int\limits_{B_{R}^{N}}}
\left\{  \text{ }%
{\displaystyle\int\limits_{\mathcal{L}_{\mathbb{R}}^{D}\left(  \mathbb{Q}%
_{p}\right)  }}
e^{\sum_{j=1}^{N}\boldsymbol{k}_{j}\cdot\boldsymbol{\varphi}\left(
x_{j}\right)  }d\mathbb{P}_{D}\left(  \boldsymbol{\varphi}\right)  \right\}
\prod_{j=1}^{N}dx_{j}\text{,} \label{Amplitude_R}%
\end{equation}
and $B_{R}^{N}$ denotes an $N$-dimensional ball of radius $p^{R}$. Following
the standard approach in QFT, we expand the right-hand side of
(\ref{Amplitude_R}) around a suitable solution of the equations of
motion.\ The main difficulty is that the solutions of these equation are
distributions, and we are restricted to work with functions. We show that
there is a change of variables in (\ref{Amplitude_R}) such that
\begin{equation}
\mathcal{A}_{R}^{\left(  N\right)  }\left(  \boldsymbol{k}\right)  =\frac
{1}{Z_{0}}%
{\displaystyle\int\limits_{B_{R}^{N}}}
\prod_{j<i}^{N}\left\vert x_{j}-x_{i}\right\vert _{p}^{2\frac{\left(
p-1\right)  }{p\ln p}\boldsymbol{k}_{i}\cdot\boldsymbol{k}_{j}}\left\{  \text{
}%
{\displaystyle\int\limits_{\mathcal{L}_{\mathbb{R}}^{D}\left(  \mathbb{Q}%
_{p}\right)  }}
e^{\sum_{j=1}^{N}\boldsymbol{k}_{j}\cdot\widetilde{\boldsymbol{\varphi}%
}\left(  x_{j}\right)  }d\widetilde{\mathbb{P}}_{D}\left(  \widetilde
{\boldsymbol{\varphi}}\right)  \right\}  \prod_{j=0}^{N}dx_{j},
\label{Amplitude_R_1}%
\end{equation}
where $\widetilde{\mathbb{P}}_{D}$ is a\ probability measure.\ Here is an
important difference with respect to the classical QFT, which is that the
$\boldsymbol{k}$ cannot be considered as a coupling constant, and thus there
is no a standard perturbative expansion for (\ref{Amplitude_R_1}). By taking
the classical \ normalization
\[
x_{1}=0\text{, }x_{N-1}=1\text{, }x_{N}=\infty\text{,}%
\]
and using the expansion of the exponential function, we show that
(\ref{Amplitude_R_1}) admits a series expansion of the form
\begin{multline*}
\mathcal{A}_{R}^{\left(  N\right)  }\left(  \boldsymbol{k}\right)
=\frac{C_{0}}{Z_{0}}\sum\limits_{l=0}^{\infty}%
{\displaystyle\int\limits_{B_{R}^{N-3}}}
\prod\limits_{i=2}^{N-2}\left\vert x_{i}\right\vert _{p}^{2\dfrac{\left(
p-1\right)  }{p\ln p}\boldsymbol{k}_{1}\cdot\boldsymbol{k}_{i}}\left\vert
1-x_{i}\right\vert _{p}^{2\dfrac{\left(  p-1\right)  }{p\ln p}\boldsymbol{k}%
_{N-1}\cdot\boldsymbol{k}_{i}}\\
\times\prod_{2\leq i,j\leq N-2}\left\vert x_{j}-x_{i}\right\vert _{p}%
^{2\dfrac{\left(  p-1\right)  }{p\ln p}\boldsymbol{k}_{i}\cdot\boldsymbol{k}%
_{j}}G_{l}(\boldsymbol{k},\boldsymbol{x})\prod_{j=2}^{N-2}dx_{j}\text{,}%
\end{multline*}
where $G_{0}(\boldsymbol{k},\boldsymbol{x})$ is a constant and the
$G_{l}(\boldsymbol{k},\boldsymbol{x})$s are continuous functions in
$\boldsymbol{x}$, for $l\geq1$. The product $1_{B_{R}^{N-3}}\left(
\boldsymbol{x}\right)  $ $G_{l}(\boldsymbol{k},\boldsymbol{x})$ can be
approximated by a test function in $\boldsymbol{x}$\ depending of
$\boldsymbol{k}$, for $l\geq1$, without altering the analytic dependence of
integral $\mathcal{A}_{R}^{\left(  N\right)  }\left(  \boldsymbol{k}\right)  $
with respect to $\boldsymbol{k}$.

An integral of the form
\begin{multline*}
Z_{\Phi}^{\left(  N\right)  }\left(  \boldsymbol{k}\right)  =%
{\displaystyle\int\limits_{\mathbb{Q}_{p}^{N-3}}}
\prod\limits_{i=2}^{N-2}\left\vert x_{i}\right\vert _{p}^{2\dfrac{\left(
p-1\right)  }{p\ln p}\boldsymbol{k}_{1}\cdot\boldsymbol{k}_{i}}\left\vert
1-x_{i}\right\vert _{p}^{2\dfrac{\left(  p-1\right)  }{p\ln p}\boldsymbol{k}%
_{N-1}\cdot\boldsymbol{k}_{i}}\\
\times\prod_{2\leq i,j\leq N-2}\left\vert x_{j}-x_{i}\right\vert _{p}%
^{2\dfrac{\left(  p-1\right)  }{p\ln p}\boldsymbol{k}_{i}\cdot\boldsymbol{k}%
_{j}}\Phi(\boldsymbol{x})\prod_{j=2}^{N-2}dx_{j}\text{,}%
\end{multline*}
where $\Phi$ is a test function is a particular case of a multivariate Igusa
zeta function \cite{Igusa}.

In \cite{Bocardo:2020mk}-\cite{Zun-B-C-JHEP} was established that the integral
$Z_{\Phi}^{\left(  N\right)  }\left(  \boldsymbol{k}\right)  $ is holomorphic
in a certain domain and that if $\Phi=1_{B_{R}^{N-3}}\left(  \boldsymbol{x}%
\right)  $, then%
\begin{multline*}
\lim_{R\rightarrow\infty}Z_{R}^{\left(  N\right)  }\left(  \boldsymbol{k}%
\right)  =%
{\displaystyle\int\limits_{\mathbb{Q}_{p}^{N-3}}}
\prod\limits_{i=2}^{N-2}\left\vert x_{i}\right\vert _{p}^{2\dfrac{\left(
p-1\right)  }{p\ln p}\boldsymbol{k}_{1}\cdot\boldsymbol{k}_{i}}\left\vert
1-x_{i}\right\vert _{p}^{2\dfrac{\left(  p-1\right)  }{p\ln p}\boldsymbol{k}%
_{N-1}\cdot\boldsymbol{k}_{i}}\\
\times\prod_{2\leq i,j\leq N-2}\left\vert x_{j}-x_{i}\right\vert _{p}%
^{2\dfrac{\left(  p-1\right)  }{p\ln p}\boldsymbol{k}_{i}\cdot\boldsymbol{k}%
_{j}}\prod_{j=2}^{N-2}dx_{j}=:Z^{\left(  N\right)  }\left(  \boldsymbol{k}%
\right)  ,
\end{multline*}
where $Z^{\left(  N\right)  }\left(  \boldsymbol{k}\right)  $ is a meromorphic
function which is a regularized version of the $p$-adic Koba-Nielsen
amplitude, \cite{Bocardo:2020mk}.

Therefore%
\begin{align*}
\mathcal{A}^{\left(  N\right)  }\left(  \boldsymbol{k}\right)   &
=\lim_{R\rightarrow\infty}\mathcal{A}_{R}^{\left(  N\right)  }\left(
\boldsymbol{k}\right)  =A^{\left(  N\right)  }\left(  \boldsymbol{k}\right) \\
&  +\lim_{R\rightarrow\infty}\left\{  \sum\limits_{l=1}^{\infty}%
{\displaystyle\int\limits_{B_{R}^{N-3}}}
\prod\limits_{i=2}^{N-2}\left\vert x_{i}\right\vert _{p}^{2\dfrac{\left(
p-1\right)  }{p\ln p}\boldsymbol{k}_{1}\cdot\boldsymbol{k}_{i}}\left\vert
1-x_{i}\right\vert _{p}^{2\dfrac{\left(  p-1\right)  }{p\ln p}\boldsymbol{k}%
_{N-1}\cdot\boldsymbol{k}_{i}}\right. \\
&  \left.  \times\frac{C_{0}}{Z_{0}}\prod_{2\leq i,j\leq N-2}\left\vert
x_{j}-x_{i}\right\vert _{p}^{2\dfrac{\left(  p-1\right)  }{p\ln p}%
\boldsymbol{k}_{i}\cdot\boldsymbol{k}_{j}}G_{l}(\boldsymbol{k},\boldsymbol{x}%
)\prod_{j=2}^{N-2}dx_{j}\right\}  ,
\end{align*}
where $A^{\left(  N\right)  }\left(  \boldsymbol{k}\right)  $ is the $p$-adic
Koba-Nielsen string amplitude in the Euclidean signature, $\frac{C_{0}}{Z_{0}%
}$ is a positive constant. We know that there is a common domain of
convergence in $\boldsymbol{k}$ for $A^{\left(  N\right)  }\left(
\boldsymbol{k}\right)  $\ and all the integrals appearing in the series, but
we do not know if the series converges. The study of the limit $R\rightarrow
\infty$ in the previous formula is an open problem.

In a forthcoming article, we plan to study the $p$-adic quantum field theories
\cite{Zuniga-Preprint} attached to a non-Archimedean version of the open
string action in a background gauge field \cite{Abouelsaood:1986gd}. This
action has cubic and quartic terms in the dynamical fields, which generate
interesting non-trivial one-loop quantum corrections which determine the beta
functions and the effective action for the gauge fields. We would like to find
the corresponding non-Archimedean version for this case. Finally, we expect
that the results presented in this work have a natural counterpart in the case
of standard Koba-Nielsen amplitudes.

\section{ Basic facts on $p$-adic analysis}

In this Section, we collect some basic results on $p$-adic analysis that we
use through the article. For a detailed exposition on $p$-adic analysis the
reader may consult \cite{Alberio et al}, \cite{Taibleson}, \cite{V-V-Z}.

\subsection{The field of $p$-adic numbers}

Throughout this article $p$ will denote a prime number. The field of $p-$adic
numbers $\mathbb{Q}_{p}$ is defined as the completion of the field of rational
numbers $\mathbb{Q}$ with respect to the $p-$adic norm $|\cdot|_{p}$, which is
defined as
\[
|x|_{p}=%
\begin{cases}
0 & \text{if }x=0\\
& \\
p^{-\gamma} & \text{if }x=p^{\gamma}\dfrac{a}{b},
\end{cases}
\]
where $a$ and $b$ are integers coprime with $p$. The integer $\gamma
=ord_{p}(x):=ord(x)$, with $ord(0):=+\infty$, is called the\textit{\ }%
$p-$\textit{adic order of} $x$. We extend the $p-$adic norm to $\mathbb{Q}%
_{p}^{N}$ by taking%
\[
||x||_{p}:=\max_{1\leq i\leq N}|x_{i}|_{p},\qquad\text{for }x=(x_{1}%
,\dots,x_{N})\in\mathbb{Q}_{p}^{N}.
\]
We define $ord(x)=\min_{1\leq i\leq N}\{ord(x_{i})\}$, then $||x||_{p}%
=p^{-ord(x)}$.\ The metric space $\left(  \mathbb{Q}_{p}^{N},||\cdot
||_{p}\right)  $ is a complete ultrametric space. As a topological space
$\mathbb{Q}_{p}$\ is homeomorphic to a Cantor-like subset of the real line,
see e.g. \cite{Alberio et al}, \cite{V-V-Z}.

Any $p-$adic number $x\neq0$ has a unique expansion of the form
\[
x=p^{ord(x)}\sum_{j=0}^{\infty}x_{j}p^{j},
\]
where $x_{j}\in\{0,1,2,\dots,p-1\}$ and $x_{0}\neq0$. By using this expansion,
we define \textit{the fractional part }$\{x\}_{p}$\textit{ of }$x\in
\mathbb{Q}_{p}$ as the rational number
\[
\{x\}_{p}=%
\begin{cases}
0 & \text{if }x=0\text{ or }ord(x)\geq0\\
& \\
p^{ord(x)}\sum_{j=0}^{-ord(x)-1}x_{j}p^{j} & \text{if }ord(x)<0.
\end{cases}
\]
In addition, any $x\in\mathbb{Q}_{p}^{N}\smallsetminus\left\{  0\right\}  $
can be represented uniquely as $x=p^{ord(x)}v\left(  x\right)  $ where
$\left\Vert v\left(  x\right)  \right\Vert _{p}=1$.

\subsection{Topology of $\mathbb{Q}_{p}^{N}$}

For $r\in\mathbb{Z}$, denote by $B_{r}^{N}(a)=\{x\in\mathbb{Q}_{p}%
^{N};||x-a||_{p}\leq p^{r}\}$ \textit{the ball of radius }$p^{r}$ \textit{with
center at} $a=(a_{1},\dots,a_{N})\in\mathbb{Q}_{p}^{N}$, and take $B_{r}%
^{N}(0):=B_{r}^{N}$. Note that $B_{r}^{N}(a)=B_{r}(a_{1})\times\cdots\times
B_{r}(a_{N})$, where $B_{r}(a_{i}):=\{x\in\mathbb{Q}_{p};|x_{i}-a_{i}|_{p}\leq
p^{r}\}$ is the one-dimensional ball of radius $p^{r}$ with center at
$a_{i}\in\mathbb{Q}_{p}$. The ball $B_{0}^{N}$ equals the product of $N$
copies of $B_{0}=\mathbb{Z}_{p}$, \textit{the ring of }$p-$\textit{adic
integers}. We also denote by $S_{r}^{N}(a)=\{x\in\mathbb{Q}_{p}^{N}%
;||x-a||_{p}=p^{r}\}$ \textit{the sphere of radius }$p^{r}$ \textit{with
center at} $a=(a_{1},\dots,a_{N})\in\mathbb{Q}_{p}^{N}$, and take $S_{r}%
^{N}(0):=S_{r}^{N}$. We notice that $S_{0}^{1}=\mathbb{Z}_{p}^{\times}$ (the
group of units of $\mathbb{Z}_{p}$), but $\left(  \mathbb{Z}_{p}^{\times
}\right)  ^{N}\subsetneq S_{0}^{N}$. The balls and spheres are both open and
closed subsets in $\mathbb{Q}_{p}^{N}$. In addition, two balls in
$\mathbb{Q}_{p}^{N}$ are either disjoint or one is contained in the other.

As a topological space $\left(  \mathbb{Q}_{p}^{N},||\cdot||_{p}\right)  $ is
totally disconnected, i.e. the only connected \ subsets of $\mathbb{Q}_{p}%
^{N}$ are the empty set and the points. A subset of $\mathbb{Q}_{p}^{N}$ is
compact if and only if it is closed and bounded in $\mathbb{Q}_{p}^{N}$, see
e.g. \cite[Section 1.3]{V-V-Z}, or \cite[Section 1.8]{Alberio et al}. The
balls and spheres are compact subsets. Thus $\left(  \mathbb{Q}_{p}%
^{N},||\cdot||_{p}\right)  $ is a locally compact topological space.

Since $(\mathbb{Q}_{p}^{N},+)$ is a locally compact topological group, there
exists a Haar measure $d^{N}x$, which is invariant under translations, i.e.
$d^{N}(x+a)=d^{N}x$. If we normalize this measure by the condition
$\int_{\mathbb{Z}_{p}^{N}}dx=1$, then $d^{N}x$ is unique.

\begin{notation}
We will use $\Omega\left(  p^{-r}||x-a||_{p}\right)  $ to denote the
characteristic function of the ball $B_{r}^{N}(a)$. For more general sets, we
will use the notation $1_{A}$ for the characteristic function of a set $A$.
\end{notation}

\subsection{The Bruhat-Schwartz space}

A complex-valued function $\varphi$ defined on $\mathbb{Q}_{p}^{N}$ is
\textit{called locally constant} if for any $x\in\mathbb{Q}_{p}^{N}$ there
exist an integer $l(x)\in\mathbb{Z}$ such that%
\begin{equation}
\varphi(x+x^{\prime})=\varphi(x)\text{ for any }x^{\prime}\in B_{l(x)}^{N}.
\label{local_constancy}%
\end{equation}
A function $\varphi:\mathbb{Q}_{p}^{N}\rightarrow\mathbb{C}$ is called a
\textit{Bruhat-Schwartz function (or a test function)} if it is locally
constant with compact support. Any test function can be represented as a
linear combination, with complex coefficients, of characteristic functions of
balls. The $\mathbb{C}$-vector space of Bruhat-Schwartz functions is denoted
by $\mathcal{D}(\mathbb{Q}_{p}^{N}):=\mathcal{D}$. We denote by $\mathcal{D}%
_{\mathbb{R}}(\mathbb{Q}_{p}^{N}):=\mathcal{D}_{\mathbb{R}}$\ the $\mathbb{R}%
$-vector space of Bruhat-Schwartz functions. For $\varphi\in\mathcal{D}%
(\mathbb{Q}_{p}^{N})$, the largest number $l=l(\varphi)$ satisfying
(\ref{local_constancy}) is called \textit{the exponent of local constancy (or
the parameter of constancy) of} $\varphi$.

We denote by $\mathcal{D}_{m}^{l}(\mathbb{Q}_{p}^{N})$ the finite-dimensional
space of test functions from $\mathcal{D}(\mathbb{Q}_{p}^{N})$ having supports
in the ball $B_{m}^{N}$ and with parameters \ of constancy $\geq l$. We now
define a topology on $\mathcal{D}$ as follows. We say that a sequence
$\left\{  \varphi_{j}\right\}  _{j\in\mathbb{N}}$ of functions in
$\mathcal{D}$ converges to zero, if the two following conditions hold:

(1) there are two fixed integers $k_{0}$ and $m_{0}$ such that \ each
$\varphi_{j}\in$ $\mathcal{D}_{m_{0}}^{k_{0}}$;

(2) $\varphi_{j}\rightarrow0$ uniformly.

$\mathcal{D}$ endowed with the above topology becomes a topological vector space.

\subsection{$L^{\rho}$ spaces}

Given $\rho\in\lbrack1,\infty)$, we denote by $L^{\rho}:=L^{\rho}\left(
\mathbb{Q}
_{p}^{N}\right)  :=L^{\rho}\left(
\mathbb{Q}
_{p}^{N},d^{N}x\right)  ,$ the $\mathbb{C}-$vector space of all the complex
valued functions $g$ satisfying%
\[%
{\displaystyle\int\limits_{\mathbb{Q}_{p}^{N}}}
\left\vert g\left(  x\right)  \right\vert ^{\rho}d^{N}x<\infty.
\]
The corresponding $\mathbb{R}$-vector spaces are denoted as $L_{\mathbb{R}%
}^{\rho}\allowbreak:=L_{\mathbb{R}}^{\rho}\left(
\mathbb{Q}
_{p}^{N}\right)  =L_{\mathbb{R}}^{\rho}\left(
\mathbb{Q}
_{p}^{N},d^{N}x\right)  $, $1\leq\rho<\infty$.

If $U$ is an open subset of $\mathbb{Q}_{p}^{N}$, $\mathcal{D}(U)$ denotes the
space of test functions with supports contained in $U$, then $\mathcal{D}(U)$
is dense in
\[
L^{\rho}\left(  U\right)  =\left\{  \varphi:U\rightarrow\mathbb{C};\left\Vert
\varphi\right\Vert _{\rho}=\left\{
{\displaystyle\int\limits_{U}}
\left\vert \varphi\left(  x\right)  \right\vert ^{\rho}d^{N}x\right\}
^{\frac{1}{\rho}}<\infty\right\}  ,
\]
where $d^{N}x$ is the normalized Haar measure on $\left(  \mathbb{Q}_{p}%
^{N},+\right)  $, for $1\leq\rho<\infty$, see e.g. \cite[Section 4.3]{Alberio
et al}. We denote by $L_{\mathbb{R}}^{\rho}\left(  U\right)  $ the real
counterpart of $L^{\rho}\left(  U\right)  $.

\subsection{The Fourier transform}

Set $\chi_{p}(y)=\exp(2\pi i\{y\}_{p})$ for $y\in\mathbb{Q}_{p}$. The map
$\chi_{p}(\cdot)$ is an additive character on $\mathbb{Q}_{p}$, i.e. a
continuous map from $\left(  \mathbb{Q}_{p},+\right)  $ into $S$ (the unit
circle considered as multiplicative group) satisfying $\chi_{p}(x_{0}%
+x_{1})=\chi_{p}(x_{0})\chi_{p}(x_{1})$, $x_{0},x_{1}\in\mathbb{Q}_{p}$. \ The
additive characters of $\mathbb{Q}_{p}$ form an Abelian group which is
isomorphic to $\left(  \mathbb{Q}_{p},+\right)  $. The isomorphism is given by
$\kappa\rightarrow\chi_{p}(\kappa x)$, see e.g. \cite[Section 2.3]{Alberio et
al}.

Given $\kappa=(\kappa_{1},\dots,\kappa_{N})$ and $y=(x_{1},\dots
,x_{N})\allowbreak\in\mathbb{Q}_{p}^{N}$, we set $\kappa\cdot x:=\sum
_{j=1}^{N}\kappa_{j}x_{j}$. The Fourier transform of $\varphi\in
\mathcal{D}(\mathbb{Q}_{p}^{N})$ is defined as
\[
(\mathcal{F}\varphi)(\kappa)=%
{\displaystyle\int\limits_{\mathbb{Q}_{p}^{N}}}
\chi_{p}(\kappa\cdot x)\varphi(x)d^{N}x\quad\text{for }\kappa\in\mathbb{Q}%
_{p}^{N},
\]
where $d^{N}x$ is the normalized Haar measure on $\mathbb{Q}_{p}^{N}$. The
Fourier transform is a linear isomorphism from $\mathcal{D}(\mathbb{Q}_{p}%
^{N})$ onto itself satisfying
\begin{equation}
(\mathcal{F}(\mathcal{F}\varphi))(\kappa)=\varphi(-\kappa), \label{Eq_FFT}%
\end{equation}
see e.g. \cite[Section 4.8]{Alberio et al}. We will also use the notation
$\mathcal{F}_{x\rightarrow\kappa}\varphi$ and $\widehat{\varphi}$\ for the
Fourier transform of $\varphi$.

The Fourier transform extends to $L^{2}$. If $f\in L^{2},$ its Fourier
transform is defined as
\[
(\mathcal{F}f)(\kappa)=\lim_{k\rightarrow\infty}%
{\displaystyle\int\limits_{||x||_{p}\leq p^{k}}}
\chi_{p}(\kappa\cdot x)f(x)d^{N}x,\quad\text{for }\kappa\in%
\mathbb{Q}
_{p}^{N},
\]
where the limit is taken in $L^{2}$. We recall that the Fourier transform is
unitary on $L^{2},$ i.e. $||f||_{2}=||\mathcal{F}f||_{2}$ for $f\in L^{2}$ and
that (\ref{Eq_FFT}) is also valid in $L^{2}$, see e.g. \cite[Chapter III,
Section 2]{Taibleson}.

\subsection{Distributions}

The $\mathbb{C}$-vector space $\mathcal{D}^{\prime}\left(  \mathbb{Q}_{p}%
^{n}\right)  $ $:=\mathcal{D}^{\prime}$ of all continuous linear functionals
on $\mathcal{D}(\mathbb{Q}_{p}^{n})$ is called the \textit{Bruhat-Schwartz
space of distributions}. Every linear functional on $\mathcal{D}$ is
continuous, i.e. $\mathcal{D}^{\prime}$\ agrees with the algebraic dual of
$\mathcal{D}$, see e.g. \cite[Chapter 1, VI.3, Lemma]{V-V-Z}. We denote by
$\mathcal{D}_{\mathbb{R}}^{\prime}\left(  \mathbb{Q}_{p}^{n}\right)  $
$:=\mathcal{D}_{\mathbb{R}}^{\prime}$ the dual space of $\mathcal{D}%
_{\mathbb{R}}$.

We endow $\mathcal{D}^{\prime}$ with the weak topology, i.e. a sequence
$\left\{  T_{j}\right\}  _{j\in\mathbb{N}}$ in $\mathcal{D}^{\prime}$
converges to $T$ if $\lim_{j\rightarrow\infty}T_{j}\left(  \varphi\right)
=T\left(  \varphi\right)  $ for any $\varphi\in\mathcal{D}$. \ The map
\[%
\begin{array}
[c]{lll}%
\mathcal{D}^{\prime}\times\mathcal{D} & \rightarrow & \mathbb{C}\\
&  & \\
\left(  T,\varphi\right)  & \rightarrow & T\left(  \varphi\right)
\end{array}
\]
is a bilinear form which is continuous in $T$ and $\varphi$ separately. We
call this map the pairing between $\mathcal{D}^{\prime}$ and $\mathcal{D}$.
From now on we will use $\left(  T,\varphi\right)  $ instead of $T\left(
\varphi\right)  $.

Every $f$\ in $L_{loc}^{1}$ defines a distribution $f\in\mathcal{D}^{\prime
}\left(  \mathbb{Q}_{p}^{N}\right)  $ by the formula
\[
\left(  f,\varphi\right)  =%
{\textstyle\int\limits_{\mathbb{Q}_{p}^{N}}}
f\left(  x\right)  \varphi\left(  x\right)  d^{N}x.
\]
Such distributions are called \textit{regular distributions}. Notice that for
$f$\ $\in L_{\mathbb{R}}^{2}$, $\left(  f,\varphi\right)  =\left\langle
f,\varphi\right\rangle $, where $\left\langle \cdot,\cdot\right\rangle $
denotes the scalar product in $L_{\mathbb{R}}^{2}$.

\subsection{The Fourier transform of a distribution}

The Fourier transform $\mathcal{F}\left[  T\right]  $ of a distribution
$T\in\mathcal{D}^{\prime}\left(  \mathbb{Q}_{p}^{n}\right)  $ is defined by%
\[
\left(  \mathcal{F}\left[  T\right]  ,\varphi\right)  =\left(  T,\mathcal{F}%
\left[  \varphi\right]  \right)  \text{ for all }\varphi\in\mathcal{D}%
(\mathbb{Q}_{p}^{n})\text{.}%
\]
The Fourier transform $T\rightarrow\mathcal{F}\left[  T\right]  $ is a linear
(and continuous) isomorphism from $\mathcal{D}^{\prime}\left(  \mathbb{Q}%
_{p}^{n}\right)  $\ onto $\mathcal{D}^{\prime}\left(  \mathbb{Q}_{p}%
^{n}\right)  $. Furthermore, $T=\mathcal{F}\left[  \mathcal{F}\left[
T\right]  \left(  -\xi\right)  \right]  $.

\section{A naive Euclidean version of the $p$-adic open string amplitudes}

We set $\boldsymbol{k}:=\left(  \boldsymbol{k}_{1},\ldots,\boldsymbol{k}%
_{N}\right)  $, where $\boldsymbol{k}_{j}=\left(  k_{0,j},\ldots
,k_{D-1,j}\right)  \in\mathbb{R}^{D}$ is the momentum of a tachyon,
$j=1,\ldots,N\text{. The dimension }D\geq1\text{ is fixed along this work. }%
$We also set%
\[
\boldsymbol{\varphi}(\cdot)=\left(  \varphi_{0}(\cdot),\ldots,\varphi
_{D-1}(\cdot)\right)  \in\left(  \mathcal{D}_{\mathbb{R}}\left(
\mathbb{Q}_{p}\right)  \right)  ^{D}.
\]
For $\boldsymbol{a}=\left(  a_{0},a_{1},\ldots,a_{D-1}\right)  $,
$\boldsymbol{b}=\text{$\left(  b_{0},\ldots,b_{D-1}\right)  $$\in$}%
\mathbb{R}^{D}$, $\boldsymbol{a}\cdot\boldsymbol{b}$ denotes the standard
scalar product in $\mathbb{R}^{D}$.

The naive Euclidean version of the $p$-adic $N$-point amplitudes is given by
\begin{equation}
\mathcal{A}^{\left(  N\right)  }\left(  \boldsymbol{k}\right)  =\frac{1}%
{Z_{0}^{phys}}\int D\boldsymbol{\varphi}\text{{}}e^{-S\left(
\boldsymbol{\varphi}\right)  }%
{\displaystyle\int\limits_{\mathbb{Q}_{p}^{N}}}
d^{N}x\text{{} }e^{\sum_{j=1}^{N}\boldsymbol{k}_{j}\cdot\boldsymbol{\varphi
}\left(  x_{j}\right)  } \label{eq:ampli}%
\end{equation}
where $d^{N}x=\prod_{j=1}^{N}dx_{j}$, $S\left(  \boldsymbol{\varphi}\right)
=\frac{T_{0}}{2}\sum_{j=0}^{D-1}S_{j}\left(  \varphi_{j}\right)  $, with
\[
S_{j}\left(  \varphi_{j}\right)  =%
{\displaystyle\int\limits_{\mathbb{Q}_{p}}}
\text{ }%
{\displaystyle\int\limits_{\mathbb{Q}_{p}}}
\genfrac{\{}{\}}{}{}{\varphi_{j}\left(  x_{j}\right)  -\varphi_{j}\left(
y_{j}\right)  }{\left\vert x_{j}-y_{j}\right\vert _{p}}%
^{2}dx_{j}dy_{j}\text{,}%
\]
and
\[
Z_{0}^{phys}=\int D\boldsymbol{\varphi}\text{{}}e^{-S\left(
\boldsymbol{\varphi}\right)  }\text{.}%
\]
The amplitudes (\ref{eq:ampli}) are just expectation values \ of products of
vertex operators. These amplitudes were proposed by Spokoiny
\cite{Spokoiny:1988zk} and Zhang \cite{Zhang}, see also \cite{Parisi},
\cite{Zabrodin:1988ep}. In these articles the authors obtain the $p$-adic open
Koba-Nielsen amplitudes from amplitudes (\ref{eq:ampli}) by a formal
calculation. The central goal of this work is to provide a mathematical
framework to understand these calculations.

Since there is $l\in\mathbb{Z}$ such that $\varphi_{j}\left(  x_{j}\right)
=0$ for $\left\vert x_{j}\right\vert _{p}>p^{l}\text{, }$for some
$l\in\mathbb{Z}$,
\[%
{\displaystyle\int\limits_{\mathbb{Q}_{p}^{N}}}
d^{N}x\text{{}}e^{\sum_{j=1}^{N}\boldsymbol{k}_{j}\cdot\boldsymbol{\varphi
}\left(  x_{j}\right)  }=\infty\text{.}%
\]
To fix this problem, it is necessary to introduce a cut-off and set
\[
\mathcal{A}_{R}^{\left(  N\right)  }\left(  \boldsymbol{k}\right)  =\frac
{1}{Z_{0}^{phys}}\int D\boldsymbol{\varphi}\text{{}}e^{-S\left(
\boldsymbol{\varphi}\right)  }%
{\displaystyle\int\limits_{B_{R}^{N}}}
d^{N}x\text{{}}e^{\sum_{j=1}^{N}\boldsymbol{k}_{j}\cdot\boldsymbol{\varphi
}\left(  x_{j}\right)  }\text{,}%
\]
where $R$ is a positive integer and $B_{R}^{N}=\left\{  x\in\mathbb{Q}_{p}%
^{N};\Vert x\Vert_{p}\leq p^{R}\right\}  \text{. }$

\subsection{The action and the Vladimirov operator}

\subsubsection{The Vladimirov operator}

The Vladimirov operator $\boldsymbol{D}:\mathcal{D}\left(  \mathbb{Q}%
_{p}\right)  \rightarrow L^{2}\left(  \mathbb{Q}_{p}\right)  $ is defined as%
\begin{align*}
\boldsymbol{D}\theta\left(  x\right)   &  =\dfrac{p^{2}}{p+1}%
{\displaystyle\int\limits_{\mathbb{Q}_{p}}}
\dfrac{\theta\left(  x\right)  -\theta\left(  y\right)  }{\left\vert
x-y\right\vert _{p}^{2}}dy=\dfrac{p^{2}}{p+1}%
{\displaystyle\int\limits_{\mathbb{Q}_{p}}}
\dfrac{\theta\left(  x\right)  -\theta\left(  x-z\right)  }{\left\vert
z\right\vert _{p}^{2}}dz\\
&  =\mathcal{F}_{\xi\rightarrow x}^{-1}\left[  \left\vert \xi\right\vert
_{p}\mathcal{F}_{x\rightarrow\xi}\theta\right]  \text{.}%
\end{align*}
This operator satisfies
\[
\boldsymbol{D}\theta\left(  x\right)  =-\frac{p^{2}}{p+1}\left\vert
x\right\vert _{p}^{-2}\ast\theta\left(  x\right)  \text{, for }\theta
\in\mathcal{D}\left(  \mathbb{Q}_{p}\right)  \text{,}%
\]
see e.g. \cite[Chapter 2, Section IX.1]{V-V-Z}.

\subsubsection{The action}

We now express the action in terms of the Vladimirov operator. For
$\varphi_{j}\in\mathcal{D}\left(  \mathbb{Q}_{p}\right)  $,
\begin{gather*}
S_{j}\left(  \varphi_{j}\right)  =%
{\displaystyle\int\limits_{\mathbb{Q}_{p}}}
\text{ }%
{\displaystyle\int\limits_{\mathbb{Q}_{p}}}
\left\{  \dfrac{\varphi_{j}\left(  x_{j}\right)  -\varphi_{j}\left(
y_{j}\right)  }{\left\vert x_{j}-y_{j}\right\vert _{p}}\right\}  ^{2}%
dx_{j}dy_{j}\\
=2%
{\displaystyle\int\limits_{\mathbb{Q}_{p}}}
\text{ }%
{\displaystyle\int\limits_{\mathbb{Q}_{p}}}
\dfrac{\varphi_{j}\left(  x_{j}\right)  \left(  \varphi_{j}\left(
x_{j}\right)  -\varphi_{j}\left(  y_{j}\right)  \right)  }{\left\vert
x_{j}-y_{j}\right\vert _{p}^{2}}dy_{j}dx_{j}=2\dfrac{\left(  p+1\right)
}{p^{2}}%
{\displaystyle\int\limits_{\mathbb{Q}_{p}}}
\varphi_{j}\left(  x_{j}\right)  \boldsymbol{D}\varphi_{j}\left(
x_{j}\right)  dx_{j}\\
=2\dfrac{p+1}{p^{2}}%
{\displaystyle\int\limits_{\mathbb{Q}_{p}}}
\varphi_{j}\left(  x_{j}\right)  \mathcal{F}_{\xi_{j}\rightarrow x_{j}}%
^{-1}\left[  \left\vert \xi_{j}\right\vert _{p}\mathcal{F}_{x_{j}%
\rightarrow\xi_{j}}\varphi_{j}\right]  dx_{j}\\
=2\dfrac{p+1}{p^{2}}%
{\displaystyle\int\limits_{\mathbb{Q}_{p}}}
\overline{\widehat{\varphi_{j}}\left(  \xi_{j}\right)  }\left\vert \xi
_{j}\right\vert _{p}\widehat{\varphi_{j}}\left(  \xi_{j}\right)  d\xi
_{j}=2\dfrac{p+1}{p^{2}}%
{\displaystyle\int\limits_{\mathbb{Q}_{p}}}
\left\vert \xi_{j}\right\vert _{p}\left\vert \widehat{\varphi_{j}}\left(
\text{$\xi$}_{j}\right)  \right\vert ^{2}d\xi_{j}\text{.}%
\end{gather*}
Then
\[
S\left(  \boldsymbol{\varphi}\right)  =\frac{T_{0}\left(  p+1\right)  }{p^{2}%
}\sum_{j=0}^{D-1}%
{\displaystyle\int\limits_{\mathbb{Q}_{p}}}
\varphi_{j}\left(  x_{j}\right)  \boldsymbol{D}\varphi_{j}\left(
x_{j}\right)  dx_{j}\text{.}%
\]

\subsubsection{The inverse of the Vladimirov operator}

We set
\[
\mathcal{L}\left(  \mathbb{Q}_{p}\right)  =\left\{  \theta\in\mathcal{D}%
\left(  \mathbb{Q}_{p}\right)  ;\widehat{\theta}\left(  0\right)  =0\right\}
\text{.}%
\]
The complex vector space $\mathcal{L}\left(  \mathbb{Q}_{p}\right)  $ endowed
with the topology inherited from $\mathcal{D}\left(  \mathbb{Q}_{p}\right)  $
is called the $p$\textit{-adic Lizorkin space of test functions of the second
kind}, see \cite[Chapter 7]{Alberio et al}. We set $\mathcal{L}_{\mathbb{R}%
}\left(  \mathbb{Q}_{p}\right)  :=\mathcal{D}_{\mathbb{R}}\left(
\mathbb{Q}_{p}\right)  \cap\mathcal{L}\left(  \mathbb{Q}_{p}\right)  $.

We define the inverse of $\mathbf{D}$ as
\[%
\begin{array}
[c]{cccc}%
\boldsymbol{D}^{-1}: & \mathcal{L}\left(  \mathbb{Q}_{p}\right)  & \rightarrow
& \mathcal{L}\left(  \mathbb{Q}_{p}\right) \\
&  &  & \\
& \theta & \rightarrow & \boldsymbol{D}^{-1}\theta\text{,}%
\end{array}
\]
where $\boldsymbol{D}^{-1}\theta\left(  x\right)  =\mathcal{F}_{\xi\rightarrow
x}^{-1}\left[  \left\vert \xi\right\vert _{p}^{-1}\mathcal{F}_{x\rightarrow
\xi}\theta\right]  $. Since $\left(  \mathcal{F}_{x\rightarrow\xi}%
\theta\right)  \left(  0\right)  =0$, we have $\boldsymbol{D}^{-1}%
\theta\left(  x\right)  \in\mathcal{L}\left(  \mathbb{Q}_{p}\right)  $.

Consider the equation
\[
\boldsymbol{D}\psi\left(  x\right)  =\theta\left(  x\right)  \text{ for
}\theta\in\mathcal{L}\left(  \mathbb{Q}_{p}\right)  \text{.}%
\]
This equation has a unique solution $\psi\in\mathcal{L}\left(  \mathbb{Q}%
_{p}\right)  $. Set
\begin{equation}
\left(  f_{1},\theta\right)  =\frac{-\left(  p-1\right)  }{p\ln p}%
{\displaystyle\int\limits_{\mathbb{Q}_{p}}}
\theta\left(  x\right)  \text{{}}\ln\left\vert x\right\vert _{p}\text{{}%
}dx\text{, \ for \ }\theta\in\mathcal{L}\left(  \mathbb{Q}_{p}\right)  .
\label{Eq_f1}%
\end{equation}
Then
\[
\widehat{f_{1}}\left(  \xi\right)  =\frac{1}{\left\vert \xi\right\vert _{p}%
}\text{\ in }\mathcal{L}^{\prime}\left(  \mathbb{Q}_{p}\right)  \text{,}%
\]
and
\[
\psi\left(  x\right)  =\boldsymbol{D}^{-1}\theta\left(  x\right)
=f_{1}\left(  x\right)  \ast\theta\left(  x\right)  \text{,}%
\]
see e.g. \cite[Chapter 2, Section IX.2]{V-V-Z}.

\section{Gaussian processes and free quantum fields}

We define the bilinear form $\mathbb{B}$
\[%
\begin{array}
[c]{cccc}%
\mathbb{B}: & \mathcal{L}_{\mathbb{R}}\left(  \mathbb{Q}_{p}\right)
\times\mathcal{L}_{\mathbb{R}}\left(  \mathbb{Q}_{p}\right)  & \rightarrow &
\mathbb{R}\\
&  &  & \\
& \left(  \varphi,\theta\right)  & \rightarrow & \langle\varphi,\boldsymbol{D}%
^{-1}\theta\rangle
\end{array}
\]
where $\langle\cdot,\cdot\rangle$ denotes the scalar product in $L^{2}\left(
\mathbb{Q}_{p}\right)  \text{.}$

\begin{lemma}
\label{Lemma1}$\mathbb{B}$ is a positive, continuous bilinear form from
$\mathcal{L}_{\mathbb{R}} \left(  \mathbb{Q}_{p}\right)  \times\mathcal{L}%
_{\mathbb{R}} \left(  \mathbb{Q}_{p}\right)  $ into $\mathbb{R}$.
\end{lemma}

\begin{proof}
We first notice that for $\varphi\in\mathcal{L}_{\mathbb{R}}\left(
\mathbb{Q}_{p}\right)  $, we have
\[
\mathbb{B}\left(  \varphi,\varphi\right)  =\langle\varphi,\boldsymbol{D}%
^{-1}\varphi\rangle=\langle\mathcal{F}^{-1}\varphi,\frac{\mathcal{F}\varphi
}{\left\vert \xi\right\vert _{p}}\rangle=%
{\displaystyle\int\limits_{\mathbb{Q}_{p}}}
\frac{\left\vert \widehat{\varphi}(\xi)\right\vert ^{2}d\xi}{\left\vert
\xi\right\vert _{p}}\geq0.
\]
Then $\mathbb{B}\left(  \varphi,\varphi\right)  =0$ implies that $\varphi$ is
zero almost everywhere and since $\varphi$ is continuous $\varphi=0$. Let
$\left(  \varphi_{n},\theta_{n}\right)  \in\mathcal{L}_{\mathbb{R}}\left(
\mathbb{Q}_{p}\right)  \times\mathcal{L}_{\mathbb{R}}\left(  \mathbb{Q}%
_{p}\right)  $ be two sequences such that $\varphi_{n}\rightarrow0$ and
$\theta_{n}\rightarrow0$ in $\mathcal{L}_{\mathbb{R}}\left(  \mathbb{Q}%
_{p}\right)  $. We recall that the topology of $\mathcal{L}_{\mathbb{R}%
}\left(  \mathbb{Q}_{p}\right)  $ agrees with the topology of $\mathcal{D}%
_{\mathbb{R}}\left(  \mathbb{Q}_{p}\right)  $. Now,%

\begin{align*}
\mathbb{B}\left(  \theta_{n},\varphi_{n}\right)   &  =%
{\displaystyle\int\limits_{\mathbb{Q}_{p}}}
\dfrac{\widehat{\theta}_{n}\left(  \xi\right)  \overline{\widehat{\varphi}%
}_{n}\left(  \xi\right)  }{\left\vert \xi\right\vert _{p}}d\xi=%
{\displaystyle\int\limits_{\mathbb{Z}_{p}}}
\dfrac{\widehat{\theta}_{n}\left(  \xi\right)  \overline{\widehat{\varphi}%
}_{n}\left(  \xi\right)  }{\left\vert \xi\right\vert _{p}}d\xi+%
{\displaystyle\int\limits_{\mathbb{Q}_{p}\smallsetminus\mathbb{Z}_{p}}}
\dfrac{\widehat{\theta}_{n}\left(  \xi\right)  \overline{\widehat{\varphi}%
}_{n}\left(  \xi\right)  }{\left\vert \xi\right\vert _{p}}d\xi\\
&  =:I_{1}\left(  \theta_{n},\varphi_{n}\right)  +I_{2}\left(  \theta
_{n},\varphi_{n}\right)  \text{.}%
\end{align*}
Since $\theta_{n}\in\mathcal{D}_{\mathbb{R}}\left(  \mathbb{Q}_{p}\right)  $
there exist two integers $m_{0},\,l_{0}$, independent of $n$, such that
\[
\text{supp }\widehat{\theta}_{n}\subset p^{l_{0}}\mathbb{Z}_{p}\text{ and
}\widehat{\theta}_{n}\left(  \xi\right)  \mid_{\xi_{0}+p^{m_{0}}\mathbb{Z}%
_{p}}=\widehat{\theta}_{n}\left(  \xi_{0}\right)
\]
for each $n\in\mathbb{N}$. Without loss of generality, we may assume that
$m_{0}$ is a positive integer. Then $\widehat{\theta}_{n}\left(  \xi\right)
\mid_{p^{m_{0}}\mathbb{Z}_{p}}=\widehat{\theta}_{n}\left(  0\right)  =0$ for
each $n\in\mathbb{N}$, and
\begin{align*}
\left\vert I_{1}\left(  \varphi_{n},\theta_{n}\right)  \right\vert  &
\leq\Vert\widehat{\varphi}_{n}\Vert_{\infty}%
{\displaystyle\int\limits_{p^{-m_{0}}<\left\vert \xi\right\vert _{p}\leq1}}
\dfrac{\left\vert \widehat{\theta}_{n}\left(  \xi\right)  \right\vert
}{\left\vert \xi\right\vert _{p}}d\xi\leq\Vert\varphi_{n}\Vert_{1}%
\Vert\widehat{\theta}_{n}\Vert_{\infty}%
{\displaystyle\int\limits_{p^{-m_{0}}<\left\vert \xi\right\vert _{p}\leq1}}
\dfrac{1}{\left\vert \xi\right\vert _{p}}d\xi\\
&  \leq C_{1}\Vert\varphi_{n}\Vert_{1}\Vert\theta_{n}\Vert_{1}\text{.}%
\end{align*}
For the second integral,
\begin{align*}
|I_{2}(\varphi_{n},\theta_{n})|  &  \leq\Vert\widehat{\varphi}_{n}%
\Vert_{\infty}%
{\displaystyle\int\limits_{\left\vert \xi\right\vert _{p}>1}}
\dfrac{\left\vert \widehat{\theta}_{n}\left(  \xi\right)  \right\vert
}{\left\vert \xi\right\vert _{p}}d\xi\leq\Vert\widehat{\varphi}_{n}%
\Vert_{\infty}%
{\displaystyle\int\limits_{\left\vert \xi\right\vert _{p}>1}}
\left\vert \widehat{\theta}_{n}\left(  \xi\right)  \right\vert d\xi\\
&  \leq\Vert\varphi_{n}\Vert_{1}\Vert\widehat{\theta}_{n}\Vert_{1}\text{.}%
\end{align*}
Therefore,
\[
\mathbb{B}\left(  \varphi_{n},\theta_{n}\right)  \leq C_{1}\Vert\varphi
_{n}\Vert_{1}\Vert\theta_{n}\Vert_{1}+\Vert\varphi_{n}\Vert_{1}\Vert
\widehat{\theta}_{n}\Vert_{1}\text{.}%
\]

Now, the continuity of $\mathbb{B}$ follows from the fact that $\varphi_{n}$
$\underrightarrow{\text{uniform.}}$ $0$ and $\theta_{n}$ $\underrightarrow
{\text{uniform.}}$ $0$ imply that $\Vert\varphi_{n}\Vert_{1}\rightarrow0$,
$\Vert\theta_{n}\Vert_{1}\rightarrow0$, and $\Vert\widehat{\theta}_{n}%
\Vert_{1}\text{{}}\rightarrow0$ as $n$ tends to infinity. The convergence of
the last sequence follows from
\[
\left\Vert \widehat{\theta_{n}}\right\Vert _{1}=%
{\displaystyle\int\limits_{p^{l_{0}}\mathbb{Z}_{p}}}
\left\vert \widehat{\theta_{n}}\left(  \xi\right)  \right\vert d\xi\leq
p^{-l_{0}}\left\Vert \widehat{\theta_{n}}\right\Vert _{\infty}\leq p^{-l_{0}%
}\left\Vert \theta_{n}\right\Vert _{1}.
\]

\end{proof}

We recall that $\mathcal{D}\left(  \mathbb{Q}_{p}\right)  $ is a nuclear space
cf. \cite[Section 4]{Bruhat}, and since any subspace of a nuclear space is
also nuclear, $\mathcal{L}_{\mathbb{R}}\left(  \mathbb{Q}_{p}\right)  $ is a
nuclear space that is dense and continuously embedded in $L_{\mathbb{R}}%
^{2}\left(  \mathbb{Q}_{p}\right)  $, cf. \cite[theorem 7.4.4]{Alberio et al}.
Then we have the following Gel'fand triple:
\[
\mathcal{L}_{\mathbb{R}}\left(  \mathbb{Q}_{p}\right)  \hookrightarrow
L_{\mathbb{R}}^{2}\left(  \mathbb{Q}_{p}\right)  \hookrightarrow
\mathcal{L}_{\mathbb{R}}^{\prime}\left(  \mathbb{Q}_{p}\right)  \text{.}%
\]
We denote by $\mathcal{B}:=\mathcal{B}\left(  \mathcal{L}_{\mathbb{R}}%
^{\prime}\left(  \mathbb{Q}_{p}\right)  \right)  $ the $\sigma-$algebra
generated by the cylinder subsets of $\mathcal{L}_{\mathbb{R}}^{\prime}\left(
\mathbb{Q}_{p}\right)  \text{.}$

Consider the mapping
\[%
\begin{array}
[c]{cccc}%
\mathcal{C}: & \mathcal{L}_{\mathbb{R}}\left(  \mathbb{Q}_{p}\right)   &
\rightarrow & \mathbb{C}\\
&  &  & \\
& f & \rightarrow & e^{\frac{-1}{2}\mathbb{B}\left(  f,f\right)  }\text{.}%
\end{array}
\]
This functional is a continuous, positive definite mapping, cf. Lemma
$\text{\ref{Lemma1},}$ and $\mathcal{C}\left(  0\right)  =1$. Then
$\mathcal{C}$ defines a characteristic functional in $\mathcal{L}_{\mathbb{R}%
}\left(  \mathbb{Q}_{p}\right)  \text{.}$ By Bochner-Minlos theorem, there
exists a unique probability measure $\mathbb{P}$ called the canonical Gaussian
measure on $\left(  \mathcal{L}_{\mathbb{R}}^{\prime}\left(  \mathbb{Q}%
_{p}\right)  ,\mathcal{B}\right)  $ given by its characteristic functional as
\begin{equation}%
{\displaystyle\int\limits_{\mathcal{L}_{\mathbb{R}}^{\prime}\left(
\mathbb{Q}_{p}\right)  }}
e^{\sqrt{-1}\left(  W,f\right)  }d\mathbb{P}\left(  W\right)  =e^{-\frac{1}%
{2}\mathbb{B}\left(  f,f\right)  },\text{\ \ }f\in\mathcal{L}_{\mathbb{R}%
}\left(  \mathbb{Q}_{p}\right)  ,\label{probmeasure}%
\end{equation}
where $\left(  \cdot,\cdot\right)  $ is the pairing between $\mathcal{L}%
_{\mathbb{R}}^{\prime}\left(  \mathbb{Q}_{p}\right)  $ and $\mathcal{L}%
_{\mathbb{R}}\left(  \mathbb{Q}_{p}\right)  $. The measure $\mathbb{P}$
corresponds to a free quantum field on $\mathcal{L}_{\mathbb{R}}^{\prime
}\left(  \mathbb{Q}_{p}\right)  $. This identification is well-known in the
Archimedean and non-Archimedean settings, see e.g. \cite[Section
6.2]{Jaffe-Glimm}, \cite[Section 5.5]{Arroyo-Zuniga}.

\section{N-point amplitudes}

\subsection{A rigorous definition of the $N$-point amplitudes}

We denote by
\[%
{\displaystyle\bigotimes\limits_{j=0}^{D-1}}
\mathbb{P}\left(  \varphi_{j}\right)  =\mathbb{P}_{D}\left(
\boldsymbol{\varphi}\right)  \text{,}%
\]
the product probability measure on the product $\sigma-$algebra $\mathcal{B}%
^{D}$. We set
\[
\mathcal{L}_{\mathbb{R}}^{D}\left(  \mathbb{Q}_{p}\right)  =\mathcal{L}%
_{\mathbb{R}}\left(  \mathbb{Q}_{p}\right)  \times\cdots\times\mathcal{L}%
_{\mathbb{R}}\left(  \mathbb{Q}_{p}\right)  \text{, }D\text{-times.}%
\]
The probability measure
\begin{equation}
\frac{1_{\mathcal{L}_{\mathbb{R}}^{D}\left(  \mathbb{Q}_{p}\right)  }\left(
\boldsymbol{\varphi}\right)  \text{{}}d\mathbb{P}_{D}\left(
\boldsymbol{\varphi}\right)  }{Z_{0}},\label{Eq_Measure}%
\end{equation}
where $Z_{0}=\int\nolimits_{\mathcal{L}_{\mathbb{R}}^{D}\left(  \mathbb{Q}%
_{p}\right)  }d\mathbb{P}_{D}\left(  \boldsymbol{\varphi}\right)  $ represents
a free quantum field in $\mathcal{L}_{\mathbb{R}}^{D}\left(  \mathbb{Q}%
_{p}\right)  $.

Intuitively, the $N$-point amplitudes are the expectation values of the
products of the vertex operators with respect to the measure (\ref{Eq_Measure}%
):%
\begin{equation}
\left\langle
{\displaystyle\prod\limits_{j=1}^{N}}
\text{ \ }%
{\displaystyle\int\limits_{\mathbb{Q}_{p}}}
dx_{j}\text{{}}e^{\boldsymbol{k}_{j}\cdot\boldsymbol{\varphi}\left(
x_{j}\right)  }\right\rangle _{\mathbb{P}_{D}}=\frac{1}{Z_{0}}%
{\displaystyle\int\limits_{\mathcal{L}_{\mathbb{R}}^{D}\left(  \mathbb{Q}%
_{p}\right)  }}
\text{ \ }%
{\displaystyle\int\limits_{\mathbb{Q}_{p}^{N}}}
d^{N}x\text{{}}e^{\sum_{j=1}^{N}\boldsymbol{k}_{j}\cdot\boldsymbol{\varphi
}\left(  x_{j}\right)  }d\mathbb{P}_{D}\left(  \boldsymbol{\varphi}\right)  .
\label{Eq_N_point_Am}%
\end{equation}
It is important to mention that $e^{\sum_{j=1}^{N}\boldsymbol{k}_{j}%
\cdot\boldsymbol{\varphi}\left(  x_{j}\right)  }$ requires that each entry of
$\boldsymbol{\varphi}\left(  x_{j}\right)  $ be a function, for this reason,
the factor $1_{\mathcal{L}_{\mathbb{R}}^{D}\left(  \mathbb{Q}_{p}\right)  }$
is completely necessary in (\ref{Eq_Measure}).

Due to the divergence of the second integral in the right-hand side of
(\ref{Eq_N_point_Am}), we define the $N$-point amplitudes as follows.

\begin{definition}
For a positive integer $R$, we define the $p$-adic $N$-point amplitudes as
$\mathcal{A}^{\left(  N\right)  }\left(  \boldsymbol{k}\right)  =\lim
_{R\rightarrow\infty}\mathcal{A}_{R}^{\left(  N\right)  }\left(
\boldsymbol{k}\right)  $, where
\[
\mathcal{A}_{R}^{\left(  N\right)  }\left(  \boldsymbol{k}\right)  :=\frac
{1}{Z_{0}}%
{\displaystyle\int\limits_{B_{R}^{N}}}
\left\{  \text{ }%
{\displaystyle\int\limits_{\mathcal{L}_{\mathbb{R}}^{D}\left(  \mathbb{Q}%
_{p}\right)  }}
e^{\sum_{j=1}^{N}\boldsymbol{k}_{j}\cdot\boldsymbol{\varphi}\left(
x_{j}\right)  }d\mathbb{P}_{D}\left(  \boldsymbol{\varphi}\right)  \right\}
\prod_{j=1}^{N}dx_{j}\text{.}%
\]

\end{definition}

Our central goal is to show (in a rigorous mathematical way) that the ansatz
proposed in the above definition allow us to obtain the $p$-adic open
Koba-Nielsen amplitudes as the constant term of a series expansion of
$\lim_{R\rightarrow\infty}\mathcal{A}_{R}^{\left(  N\right)  }\left(
\boldsymbol{k}\right)  $ in functions depending on $\boldsymbol{k}$. The
precise statements of our main results are given in Theorems \ref{Prop2},
\ref{Prop3}.

By using that
\[
\sum_{j=1}^{N}\boldsymbol{k}_{j}\cdot\boldsymbol{\varphi}\left(  x_{j}\right)
=\sum_{j=1}^{N}\sum_{l=0}^{D-1}k_{l,j}\varphi_{l}\left(  x_{j}\right)
\]
we have
\begin{align}
\mathcal{A}_{R}^{\left(  N\right)  }\left(  \boldsymbol{k}\right)   &
=\dfrac{1}{Z_{0}}%
{\displaystyle\int\limits_{B_{R}^{N}}}
\left\{  \text{ }%
{\displaystyle\int\limits_{\mathcal{L}_{\mathbb{R}}^{D}\left(  \mathbb{Q}%
_{p}\right)  }}
e^{\sum_{j=1}^{N}\sum_{l=0}^{D-1}k_{l,j}\varphi_{l}\left(  x_{j}\right)
}\prod_{l=0}^{D-1}d\mathbb{P}\left(  \varphi_{l}\right)  \right\}  \prod
_{j=1}^{N}dx_{j}\nonumber\\
&  =\dfrac{1}{Z_{0}}%
{\displaystyle\int\limits_{B_{R}^{N}}}
\left\{  \prod_{l=0}^{D-1}\text{ }%
{\displaystyle\int\limits_{\mathcal{L}_{\mathbb{R}}\left(  \mathbb{Q}%
_{p}\right)  }}
e^{\sum_{j=1}^{N}k_{l,j}\varphi_{l}\left(  x_{j}\right)  }d\mathbb{P}\left(
\varphi_{l}\right)  \right\}  \prod_{j=1}^{N}dx_{j}. \label{eq:amplfinal}%
\end{align}
We now introduce the notation
\begin{equation}
\sum_{j=1}^{N}k_{l,j}\varphi_{l}\left(  x_{j}\right)  :=\sum_{j=1}^{N}%
v_{j}\varphi\left(  x_{j}\right)  \label{Notation}%
\end{equation}
taking advantage that $l$ is fixed. Here $v_{j}\in\mathbb{R}$ and $\varphi
\in\mathcal{L}_{\mathbb{R}}\left(  \mathbb{Q}_{p}\right)  $.

We set
\begin{equation}
\widetilde{\mathcal{A}}_{R}^{\left(  N\right)  }\left(  \boldsymbol{x}%
,\boldsymbol{v}\right)  :=\frac{1}{Z_{0}^{1/D}}%
{\displaystyle\int\limits_{\mathcal{L}_{\mathbb{R}}\left(  \mathbb{Q}%
_{p}\right)  }}
e^{\sum_{j=1}^{N}v_{j}\varphi\left(  x_{j}\right)  }d\mathbb{P}\left(
\varphi\right)  , \label{eq:aprox1}%
\end{equation}
where $\boldsymbol{x}=\left(  x_{1},\ldots,x_{N}\right)  \in\mathbb{Q}_{p}%
^{N}$, $\boldsymbol{v}=\left(  v_{1},\ldots,v_{N}\right)  \in\mathbb{R}^{N}$.
Notice that
\[
\sum_{j=1}^{N}v_{j}\varphi\left(  x_{j}\right)  =\sum_{j=1}^{N}v_{j}\left(
\delta\left(  x-x_{j}\right)  ,\varphi\left(  x\right)  \right)  \text{,}%
\]
where $\delta\left(  \cdot-x_{j}\right)  $ denotes the Dirac distribution
centered at $x_{j}$.

\begin{lemma}
\label{Lemma2}$\widetilde{\mathcal{A}}_{R}^{\left(  N\right)  }\left(
\boldsymbol{x},\boldsymbol{v}\right)  <\infty$ for any $R,N,\boldsymbol{x}%
,\boldsymbol{v}$. Furthermore, for $R,N,\boldsymbol{v}$ fixed, $\widetilde
{\mathcal{A}}_{R}^{\left(  N\right)  }\left(  \boldsymbol{x},\boldsymbol{v}%
\right)  $ is a continuous function in $\boldsymbol{x}$.
\end{lemma}

\begin{proof}
We first recall that
\begin{equation}
{\int\limits_{\mathcal{L}_{\mathbb{R}}^{\prime}\left(  \mathbb{Q}_{p}\right)
}}e^{\left(  W,\theta\right)  }d\mathbb{P}\left(  W\right)  <\infty
\label{Eq_White_Noise}%
\end{equation}
for any $\theta\in\mathcal{L}_{\mathbb{R}}\left(  \mathbb{Q}_{p}\right)  $,
here $\left(  W,\theta\right)  $ denotes the pairing between the space of
distributions $\mathcal{L}_{\mathbb{R}}^{\prime}\left(  \mathbb{Q}_{p}\right)
$ and the space of Lizorkin test functions $\mathcal{L}_{\mathbb{R}}\left(
\mathbb{Q}_{p}\right)  $, cf. \cite[Theorem 1.7]{hidawhitenoise}.

By using that $\sum_{j=1}^{N}$ $\left\vert v_{j}\right\vert \left\vert
\varphi\left(  x_{j}\right)  \right\vert \delta\left(  x-x_{j}\right)
\in\mathcal{L}_{\mathbb{R}}^{\prime}\left(  \mathbb{Q}_{p}\right)  $, for any
$\varphi\in\mathcal{L}_{\mathbb{R}}\left(  \mathbb{Q}_{p}\right)  $, and by
fixing $\theta\in\mathcal{L}_{\mathbb{R}}\left(  \mathbb{Q}_{p}\right)  $ such
that $\theta\left(  x_{j}\right)  >1$ for $j=1,\ldots,N$, we have
\begin{gather*}
\sum_{j=1}^{N}v_{j}\varphi\left(  x_{j}\right)  \leq\sum_{j=1}^{N}\left\vert
v_{j}\right\vert \left\vert \varphi\left(  x_{j}\right)  \right\vert \leq
\sum_{j=1}^{N}\left\vert v_{j}\right\vert \left\vert \varphi\left(
x_{j}\right)  \right\vert \theta\left(  x_{j}\right) \\
=\left(  \sum_{j=1}^{N}\left\vert v_{j}\right\vert \left\vert \varphi\left(
x_{j}\right)  \right\vert \delta\left(  x-x_{j}\right)  ,\theta\left(
x\right)  \right)  ,
\end{gather*}
and thus
\begin{multline*}
{\int\limits_{\mathcal{L}_{\mathbb{R}}\left(  \mathbb{Q}_{p}\right)  }}%
e^{\sum_{j=1}^{N}v_{j}\varphi\left(  x_{j}\right)  }d\mathbb{P}\left(
\varphi\right)  \leq{\int\limits_{\mathcal{L}_{\mathbb{R}}\left(
\mathbb{Q}_{p}\right)  }}e^{\sum_{j=1}^{N}\left\vert v_{j}\right\vert
\left\vert \varphi\left(  x_{j}\right)  \right\vert }d\mathbb{P}\left(
\varphi\right) \\
\leq{\int\limits_{\mathcal{L}_{\mathbb{R}}\left(  \mathbb{Q}_{p}\right)  }%
e}^{\left(  \sum_{j=1}^{N}\left\vert v_{j}\right\vert \left\vert
\varphi\left(  x_{j}\right)  \right\vert \delta\left(  x-x_{j}\right)
,\theta\left(  x\right)  \right)  }d\mathbb{P}\left(  \varphi\right)
\leq{\int\limits_{\mathcal{L}_{\mathbb{R}}^{\prime}\left(  \mathbb{Q}%
_{p}\right)  }}e^{\left(  W,\theta\right)  }d\mathbb{P}\left(  W\right)
<\infty
\end{multline*}

Finally, the continuity in $\boldsymbol{x}$ follows from the dominated
convergence theorem by using that%
\[%
{\displaystyle\int\limits_{\mathcal{L}_{\mathbb{R}}^{\prime}\left(
\mathbb{Q}_{p}\right)  }}
1_{\mathcal{L}_{\mathbb{R}}\left(  \mathbb{Q}_{p}\right)  }\left(
\varphi\right)  e^{\sum_{j=1}^{N}v_{j}\varphi\left(  x_{j}\right)
}d\mathbb{P}\left(  \varphi\right)  \leq{\int\limits_{\mathcal{L}_{\mathbb{R}%
}^{\prime}\left(  \mathbb{Q}_{p}\right)  }}e^{\left(  W,\theta\right)
}d\mathbb{P}\left(  W\right)  .
\]

\end{proof}

\begin{corollary}
\label{Cor1}For $R$ fixed, $\mathcal{A}_{R}^{\left(  N\right)  }\left(
\boldsymbol{k}\right)  <\infty$ for any $\boldsymbol{k}$. Furthermore,
\begin{align*}
\mathcal{A}_{R}^{\left(  N\right)  }\left(  \boldsymbol{k}\right)   &
=\dfrac{1}{Z_{0}}%
{\displaystyle\int\limits_{B_{R}^{N}}}
\left\{  \text{ }%
{\displaystyle\int\limits_{\mathcal{L}_{\mathbb{R}}^{D}\left(  \mathbb{Q}%
_{p}\right)  }}
e^{\sum_{j=1}^{N}\boldsymbol{k}_{j}\cdot\boldsymbol{\varphi}\left(
x_{j}\right)  }d\mathbb{P}_{D}\left(  \boldsymbol{\varphi}\right)  \right\}
\prod_{j=1}^{N}dx_{j}\\
&  =\dfrac{1}{Z_{0}}%
{\displaystyle\int\limits_{\mathcal{L}_{\mathbb{R}}^{D}\left(  \mathbb{Q}%
_{p}\right)  }}
\left\{  \text{ }%
{\displaystyle\int\limits_{B_{R}^{N}}}
e^{\sum_{j=1}^{N}\boldsymbol{k}_{j}\cdot\boldsymbol{\varphi}\left(
x_{j}\right)  }\prod_{j=1}^{N}dx_{j}\right\}  d\mathbb{P}_{D}\left(
\boldsymbol{\varphi}\right)  \text{.}%
\end{align*}

\end{corollary}

\begin{proof}
By Lemma \ref{Lemma2}, for $R,N,\boldsymbol{k}$ given,
\[
\boldsymbol{x\rightarrow}\prod_{l=0}^{D-1}%
{\displaystyle\int\limits_{\mathcal{L}_{\mathbb{R}}\left(  \mathbb{Q}%
_{p}\right)  }}
e^{\sum_{j=1}^{N}k_{l,j}\varphi_{l}\left(  x_{j}\right)  }d\mathbb{P}\left(
\varphi_{l}\right)  =%
{\displaystyle\int\limits_{\mathcal{L}_{\mathbb{R}}^{D}\left(  \mathbb{Q}%
_{p}\right)  }}
e^{\sum_{j=1}^{N}\sum_{l=0}^{D-1}k_{l,j}\varphi_{l}\left(  x_{j}\right)
}\prod_{l=0}^{D-1}d\mathbb{P}\left(  \varphi_{l}\right)  <\infty
\]
is a well-defined and continuous function. Now, the announced formula is a
consequence of Fubini's theorem.
\end{proof}

\subsection{Some technical results}

We set
\[
\delta_{n}\left(  x\right)  =%
\begin{cases}
p^{n} & \left\vert x\right\vert _{p}\leq p^{-n}\\
0 & \left\vert x\right\vert _{p}>p^{-n}\text{,}%
\end{cases}
\]
for a positive integer $n$, and recall that $\delta_{n}\left(  x\right)  $
$\underrightarrow{\mathcal{D}^{\prime}\left(  \mathbb{Q}_{p}\right)  }$
$\delta\left(  x\right)  $, the Dirac distribution, as $n\rightarrow\infty$.
We now introduce an approximation for $\widetilde{\mathcal{A}}_{R}^{\left(
N\right)  }\left(  \boldsymbol{x},\boldsymbol{v}\right)  $ given by
\begin{equation}
\widetilde{\mathcal{A}}_{R}^{\left(  N\right)  }\left(  \boldsymbol{x}%
,\boldsymbol{v};I\right)  :=\frac{1}{Z_{0}^{1/D}}%
{\displaystyle\int\limits_{\mathcal{L}_{\mathbb{R}}\left(  \mathbb{Q}%
_{p}\right)  }}
e^{\sum_{j=1}^{N}v_{j}\left(  \delta_{I}\left(  x-x_{j}\right)  ,\varphi
\left(  x\right)  \right)  }d\mathbb{P}\left(  \varphi\right)  ,
\label{eq:aprox2}%
\end{equation}
where $I$ is a positive integer.

\begin{lemma}
\label{Lemma3}%
\[
\lim_{I\rightarrow\infty}\text{{}}\widetilde{\mathcal{A}}_{R}^{\left(
N\right)  }\left(  \boldsymbol{x},\boldsymbol{v};I\right)  =\widetilde
{\mathcal{A}}_{R}^{\left(  N\right)  }\left(  \boldsymbol{x},\boldsymbol{v}%
\right)  \text{.}%
\]

\end{lemma}

\begin{proof}
The proof is similar to the one given for Lemma \ref{Lemma2}. The result
follows from%
\[
1_{\mathcal{L}_{\mathbb{R}}\left(  \mathbb{Q}_{p}\right)  }\left(
\varphi\right)  e^{\sum_{j=1}^{N}v_{j}\left(  \delta_{I}\left(  z-x_{j}%
\right)  ,\varphi\left(  x\right)  \right)  }\leq1_{\mathcal{L}_{\mathbb{R}%
}\left(  \mathbb{Q}_{p}\right)  }\left(  \varphi\right)  e^{\left(
W,\theta\right)  },
\]
where $W\in\mathcal{L}_{\mathbb{R}}^{\prime}\left(  \mathbb{Q}_{p}\right)  $
is distribution depending on $x_{j}$, $v_{j}$, for $\ j=1,\ldots,N$, but \ not
on $I$, and where $\theta\in\mathcal{L}_{\mathbb{R}}\left(  \mathbb{Q}%
_{p}\right)  $ is a fixed \ positive function. Now, by using
(\ref{Eq_White_Noise}) and the dominated convergence theorem, we have%

\[
\left\vert \sum_{j=1}^{N}v_{j}\left(  \delta_{I}\left(  z-x_{j}\right)
,\varphi\left(  z\right)  \right)  \right\vert =\left\vert p^{I}\sum_{j=1}%
^{N}v_{j}%
{\displaystyle\int\limits_{x_{j}+p^{I}\mathbb{Z}_{p}}}
\varphi\left(  y\right)  dy\right\vert \leq p^{I}\sum_{j=1}^{N}\left\vert
v_{j}\right\vert
{\displaystyle\int\limits_{x_{j}+p^{I}\mathbb{Z}_{p}}}
\left\vert \varphi\left(  y\right)  \right\vert dy.
\]
We denote by $l_{\varphi}$ index of local constancy of $\varphi$. We pick
$I_{\varphi}=\max\left\{  I,l_{\varphi}\right\}  $, then $p^{I_{\varphi}%
}\mathbb{Z}_{p}$ is a subgroup of $p^{I}\mathbb{Z}_{p}$ and
\[
G_{j}:=\left(  x_{j}+p^{I}\mathbb{Z}_{p}\right)  /p^{I_{\varphi}}%
\mathbb{Z}_{p}%
\]
is a finite set such that $x_{j}+p^{I}\mathbb{Z}_{p}=%
{\textstyle\bigsqcup\nolimits_{\widetilde{x}\in G_{j}}}
\left(  \widetilde{x}+p^{I_{\varphi}}\mathbb{Z}_{p}\right)  $ (disjoint
union). We pick a function $\theta\in\mathcal{L}_{\mathbb{R}}\left(
\mathbb{Q}_{p}\right)  $ satisfying $\theta\left(  \widetilde{x}\right)
\geq1$ for $\widetilde{x}\in G_{j}$, now
\begin{gather*}
p^{I}\sum_{j=1}^{N}\left\vert v_{j}\right\vert
{\displaystyle\int\limits_{x_{j}+p^{I}\mathbb{Z}_{p}}}
\left\vert \varphi\left(  y\right)  \right\vert dy=p^{I}\sum_{j=1}^{N}%
{\displaystyle\sum\limits_{\widetilde{x}\in G_{j}}}
\left\vert v_{j}\right\vert
{\displaystyle\int\limits_{\widetilde{x}+p^{I_{\varphi}}\mathbb{Z}_{p}}}
\left\vert \varphi\left(  y\right)  \right\vert dy\\
=p^{I-I_{\varphi}}\sum_{j=1}^{N}%
{\displaystyle\sum\limits_{\widetilde{x}\in G_{j}}}
\left\vert v_{j}\right\vert \left\vert \varphi\left(  \widetilde{x}\right)
\right\vert \leq\sum_{j=1}^{N}%
{\displaystyle\sum\limits_{\widetilde{x}\in G_{j}}}
\left\vert v_{j}\right\vert \left\vert \varphi\left(  \widetilde{x}\right)
\right\vert \\
\leq\sum_{j=1}^{N}%
{\displaystyle\sum\limits_{\widetilde{x}\in G_{j}}}
\left\vert v_{j}\right\vert \left\vert \varphi\left(  \widetilde{x}\right)
\right\vert \theta\left(  \widetilde{x}\right)  =\sum_{j=1}^{N}%
{\displaystyle\sum\limits_{\widetilde{x}\in G_{j}}}
\left\vert v_{j}\right\vert \left(  \left\vert \varphi\left(  \widetilde
{x}\right)  \right\vert \delta\left(  z-\widetilde{x}\right)  ,\theta\left(
z\right)  \right)  .
\end{gather*}

\end{proof}

\subsubsection{A change of variables}

Let $\varphi_{L,m}$, $\widetilde{\varphi}$ be functions in $\mathcal{L}%
_{\mathbb{R}}\left(  \mathbb{Q}_{p}\right)  $, for $L\geq1$ and $m\in
\mathbb{Q}_{p}^{\times}$. We now use the measurable mapping
\[%
\begin{array}
[c]{ccc}%
\mathcal{L}_{\mathbb{R}}\left(  \mathbb{Q}_{p}\right)  & \rightarrow &
\mathcal{L}_{\mathbb{R}}\left(  \mathbb{Q}_{p}\right) \\
&  & \\
\widetilde{\varphi}-\varphi_{L,m} & \rightarrow & \varphi\text{,}%
\end{array}
\]
as a change of variables in (\ref{eq:aprox2}). There exist a measure
$\widetilde{\mathbb{P}}_{L,m}$ such that
\begin{align}
\widetilde{\mathcal{A}}_{R}^{\left(  N\right)  }\left(  \boldsymbol{x}%
,\boldsymbol{v};I\right)   &  =\dfrac{1}{Z_{0}^{1/D}}%
{\displaystyle\int\limits_{\mathcal{L}_{\mathbb{R}}\left(  \mathbb{Q}%
_{p}\right)  }}
e^{\sum_{j=1}^{N}v_{j}\left(  \delta_{I}\left(  x-x_{j}\right)  ,\widetilde
{\varphi}-\varphi_{L,m}\right)  }d\widetilde{\mathbb{P}}_{L,m}\left(
\widetilde{\varphi}\right) \nonumber\\
&  =\dfrac{1}{Z_{0}^{1/D}}\text{ }e^{\sum_{j=1}^{N}v_{j}\left(  \delta
_{I}\left(  x-x_{j}\right)  ,-\varphi_{L,m}\right)  }%
{\displaystyle\int\limits_{\mathcal{L}_{\mathbb{R}}\left(  \mathbb{Q}%
_{p}\right)  }}
e^{\sum_{j=1}^{N}v_{j}\left(  \delta_{I}\left(  x-x_{j}\right)  ,\widetilde
{\varphi}\right)  }d\widetilde{\mathbb{P}}_{L,m}\left(  \widetilde{\varphi
}\right)  . \label{eq:change of variab}%
\end{align}

Our next goal is to compute the limits $\left\vert m\right\vert _{p}%
\rightarrow\infty$, $L\rightarrow\infty$, $I\rightarrow\infty$ in
(\ref{eq:change of variab}) to obtain a formula for $\mathcal{A}_{R}^{\left(
N\right)  }\left(  \boldsymbol{k}\right)  $. This calculation is carried out
in two steps.

\subsubsection{Calculation of the first limit}

Define for $L\geq1$ and $m\in\mathbb{Q}_{p}^{\times}$,
\[
J_{L,m}\left(  x\right)  =\sum_{j=1}^{N}v_{j}\delta_{L}\left(  x-x_{j}\right)
-\sum_{j=1}^{N}v_{j}\left\vert m\right\vert _{p}^{-1}\Omega\left(
\frac{\left\vert x\right\vert _{p}}{\left\vert m\right\vert _{p}}\right)
\ast\delta_{L}\left(  x-x_{j}\right)  \text{.}%
\]

\begin{lemma}
\label{lemmaJ_l,m} With the above notation, the following holds true:

\noindent(i) $J_{L ,m} \left(  x\right)  \in\mathcal{L}_{\mathbb{R}} \left(
\mathbb{Q}_{p}\right)  $ for any $L \geq1\text{,}$ $m \in\mathbb{Q}_{p}^{
\times}$;

\noindent(ii) $J_{L,m}\left(  x\right)  \rightarrow J_{L}\left(  x\right)
:=\sum_{j=1}^{N}v_{j}\delta_{L}\left(  x-x_{j}\right)  $ in $L^{\rho}\left(
\mathbb{Q}_{p}\right)  $,$\,1<\rho<\infty$ as $\left\vert m\right\vert
_{p}\rightarrow\infty$;

\noindent(iii) $J_{L,m}\left(  x\right)  \rightarrow J_{L}\left(  x\right)  $
in $\mathcal{L}_{\mathbb{R}}^{\prime}\left(  \mathbb{Q}_{p}\right)  $ as
$\left\vert m\right\vert _{p}\rightarrow\infty$;

\noindent(iv) The equation $\boldsymbol{D} \varphi_{L ,m} =J_{L ,m}$ has a
unique solution $\varphi_{L ,m} \in\mathcal{L}_{\mathbb{R}} \left(
\mathbb{Q}_{p}\right)  $ given by $\varphi_{L ,m} =f_{1} \ast J_{L ,m}$, where
$f_{1}$ is defined in (\ref{Eq_f1});

\noindent(v) $\varphi_{L,m}\rightarrow f_{1}\ast J_{L}$ in $\mathcal{L}%
_{\mathbb{R}}^{\prime}\left(  \mathbb{Q}_{p}\right)  $ as $\left\vert
m\right\vert _{p}\rightarrow\infty$;

\noindent(vi) $f_{1}\ast J_{L}=\frac{1-p}{p\ln p}\sum_{j=1}^{N}v_{j}%
\ln\left\vert x-x_{j}\right\vert _{p}$, if $\left\vert x-x_{j}\right\vert
_{p}>p^{-L}$ for $j=1,\ldots,N$.
\end{lemma}

\begin{proof}
(i) Denote by $\Delta_{m}\left(  \xi\right)  =\Omega\left(  \left\vert
m\xi\right\vert _{p}\right)  $, $m\in\mathbb{Q}_{p}^{\times}$, the
characteristic function of the ball $B_{\log_{p}\left\vert m\right\vert
_{p}^{-1}}$. Then
\[
\widehat{J}_{L,m}\left(  \xi\right)  =\sum_{j=1}^{N}v_{j}\chi_{p}\left(
\xi\cdot x_{j}\right)  \Delta_{L}\left(  \xi\right)  \left(  1-\Delta
_{m}\left(  \xi\right)  \right)  \text{,}%
\]
where $\Delta_{L}\left(  \xi\right)  =\Omega\left(  p^{-L}\left\vert
\xi\right\vert _{p}\right)  $, which implies that $\widehat{J}_{L,m}$ is a
test function satisfying $\widehat{J}_{L,m}\left(  0\right)  =0$ for
$\left\vert m\right\vert _{p}>1$.

(ii) Notice that
\[
J_{L,m}\left(  x\right)  -\sum_{j=1}^{N}v_{j}\delta_{L}\left(  x-x_{j}\right)
=-\sum_{j=1}^{N}v_{j}\left\vert m\right\vert _{p}^{-1}\Omega\left(
\dfrac{\left\vert x\right\vert _{p}}{\left\vert m\right\vert _{p}}\right)
\ast\delta_{L}\left(  x-x_{j}\right)  .
\]
By using that $\left\vert m\right\vert _{p}^{-1}\Omega\left(  \left\vert
m\right\vert _{p}^{-1}\left\vert x\right\vert _{p}\right)  \in L^{1}\left(
\mathbb{Q}_{p}\right)  $ and $\delta_{L}\left(  x\right)  \in L^{\rho}$,
$1<\rho<\infty$, and applying \cite[Lemma 7.4.2]{Alberio et al} we have
$\Omega\left(  \dfrac{\left\vert x\right\vert _{p}}{\left\vert m\right\vert
_{p}}\right)  \ast\delta_{L}\left(  x-x_{j}\right)  $ $\underrightarrow
{L^{\rho}}$\ $0$ as $\left\vert m\right\vert _{p}\rightarrow\infty$.

(iii) Take $\theta\in\mathcal{L}_{\mathbb{R}}\left(  \mathbb{Q}_{p}\right)  $,
by using the Cauchy--Schwarz inequality,
\begin{align*}
\left\vert \int\nolimits_{\mathbb{Q}_{p}}J_{L,m}\left(  x\right)
\theta\left(  x\right)  dx-\int\nolimits_{\mathbb{Q}_{p}}J_{L}\left(
x\right)  \theta\left(  x\right)  dx\right\vert  &  =\left\vert \int
\nolimits_{\mathbb{Q}_{p}}\theta\left(  x\right)  \left(  J_{L,m}\left(
x\right)  -J_{L}\left(  x\right)  \right)  dx\right\vert \\
&  \leq\Vert\theta\Vert_{2}\Vert J_{L,m}-J_{L}\Vert_{2}\text{.}%
\end{align*}
By the second part $\Vert J_{L,m}-J_{L}\Vert_{2}\rightarrow0$ as $\left\vert
m\right\vert _{p}\rightarrow\infty$.

(iv) See \cite[Chapter 2, Section IX.2]{V-V-Z} \ or \cite[Theorem
9.2.6]{Alberio et al}. (v) It follows from the third part by using the
continuity of the convolution.

(vi) If $\left\vert x-x_{j}\right\vert _{p}>p^{-L}$ for any $j=1,\ldots,N$,
\begin{align}
f_{1}\ast J_{L}\left(  x\right)   &  =\dfrac{1-p}{p\ln p}\sum_{j=1}^{N}%
v_{j}\ln\left\vert x\right\vert _{p}\ast\delta_{L}\left(  x-x_{j}\right)
\nonumber\\
&  =\dfrac{1-p}{p\ln p}\sum_{j=1}^{N}v_{j}p^{L}%
{\displaystyle\int\limits_{x-x_{j}+p^{L}\mathbb{Z}_{p}}}
\ln\left\vert z\right\vert _{p}dz=\dfrac{1-p}{p\ln p}\sum_{j=1}^{N}v_{j}%
\ln\left\vert x-x_{j}\right\vert _{p}\text{.} \label{Eq_Formula_Correcta}%
\end{align}

\end{proof}

\begin{lemma}
\label{Lemma_A}%
\[
\lim_{I\rightarrow\infty}\text{ }\lim_{L\rightarrow\infty}\text{ }%
\lim_{\left\vert m\right\vert _{p}\rightarrow\infty}e^{\sum_{j=1}^{N}%
v_{j}\left(  \delta_{I}\left(  x-x_{j}\right)  ,-\varphi_{L,m}\right)
}=e^{\dfrac{p-1}{p\ln p}\sum_{j=1}^{N}\sum_{i=1,i\neq j}^{N}v_{j}v_{i}%
\ln\left\vert x_{j}-x_{i}\right\vert _{p}}.
\]

\end{lemma}

\begin{proof}
By using the formula for $f_{1}\ast J_{L}\left(  x\right)  $, in the case
$\left\vert x-x_{j}\right\vert _{p}>p^{-L}$ for any $x_{j}$, see
(\ref{Eq_Formula_Correcta}), and the continuity of the pairing and the
continuity of the convolution,
\[
\sum_{j=1}^{N}v_{j}\left(  \delta_{I}\left(  x-x_{j}\right)  ,-\varphi
_{L,m}\right)  \rightarrow\sum_{j=1}^{N}v_{j}\left(  \delta_{I}\left(
x-x_{j}\right)  ,\frac{p-1}{p\ln p}\sum_{i=1}^{N}v_{i}\ln\left\vert
x-x_{i}\right\vert _{p}\right)
\]
in $\mathcal{L}_{\mathbb{R}}^{\prime}\left(  \mathbb{Q}_{p}\right)  $ as
$\left\vert m\right\vert _{p}\rightarrow\infty$, which implies that
\[
e^{\sum_{j=1}^{N}v_{j}(\delta_{I}\left(  x-x_{j}\right)  ,-\varphi_{L,m}%
)}\rightarrow e^{\sum_{j=1}^{N}v_{j}\left(  \delta_{I}\left(  x-x_{j}\right)
,\frac{p-1}{p\ln p}\sum_{i=1}^{N}v_{i}\ln\left\vert x-x_{i}\right\vert
_{p}\right)  }%
\]
as $\left\vert m\right\vert _{p}\rightarrow\infty$. Now since $\ln\left\vert
x\right\vert _{p}$ is locally constant in $\mathbb{Q}_{p}^{\times}$, and the
$\lim_{t\rightarrow-\infty}e^{t}=0$, we have for $I$ sufficiently large that
\[
\left(  \delta_{I}\left(  x-x_{j}\right)  ,\frac{p-1}{p\ln p}\sum_{i=1}%
^{N}v_{i}\ln\left\vert x-x_{i}\right\vert _{p}\right)  =\frac{p-1}{p\ln p}%
\sum_{i=1}^{N}v_{i}\ln\left\vert x_{j}-x_{i}\right\vert _{p},
\]
if $x_{j}\neq x_{i}$, and $-\infty$ otherwise. Therefore,
\[
e^{\sum_{j=1}^{N}v_{j}\left(  \delta_{I}\left(  x-x_{j}\right)  ,\frac
{p-1}{p\ln p}\sum_{i=1}^{N}v_{i}\ln\left\vert x-x_{i}\right\vert _{p}\right)
}=e^{\frac{p-1}{p\ln p}\sum_{j=1}^{N}\sum_{i=1,i\neq j}^{N}v_{j}v_{i}%
\ln\left\vert x_{j}-x_{i}\right\vert _{p}},
\]
for $I$ sufficiently large.
\end{proof}

\subsubsection{Calculation of the second limit}

We now describe the measure $\widetilde{\mathbb{P}}_{L,m}$. Take
$\varphi_{L,m}\in\mathcal{L}_{\mathbb{R}}^{\prime}\left(  \mathbb{Q}%
_{p}\right)  $, $\widetilde{W}\in\mathcal{L}_{\mathbb{R}}^{\prime}\left(
\mathbb{Q}_{p}\right)  $, by using (\ref{probmeasure}) and changing variables
as $W=\widetilde{W}-\varphi_{L,m}$, we have
\[%
{\displaystyle\int\limits_{\mathcal{L}_{\mathbb{R}}^{\prime}\left(
\mathbb{Q}_{p}\right)  }}
e^{\sqrt{-1}\left(  \widetilde{W}-\varphi_{L,m},g\right)  }d\widetilde
{\mathbb{P}}_{L,m}\left(  \widetilde{W}\right)  =e^{-\frac{1}{2}%
\mathbb{B}\left(  g,g\right)  }\text{,}%
\]
i.e.
\begin{equation}%
{\displaystyle\int\limits_{\mathcal{L}_{\mathbb{R}}^{\prime}\left(
\mathbb{Q}_{p}\right)  }}
e^{\sqrt{-1}\left(  \widetilde{W},g\right)  }d\widetilde{\mathbb{P}}%
_{L,m}\left(  \widetilde{W}\right)  =e^{\sqrt{-1}\left(  \varphi
_{L,m},g\right)  -\frac{1}{2}\mathbb{B}\left(  g,g\right)  }=:\mathcal{C}%
_{L,m}\left(  g\right)  \label{eq:newmeasure}%
\end{equation}
Notice that by Lemma \ref{lemmaJ_l,m},
\[
\lim_{L\rightarrow\infty}\text{ }\lim_{\left\vert m\right\vert _{p}%
\rightarrow\infty}\ \mathcal{C}_{L,m}\left(  g\right)  =e^{\sqrt{-1}\left(
\frac{p-1}{p\ln p}\sum_{j=1}^{N}v_{j}\ln\left\vert x-x_{j}\right\vert
_{p},g\right)  -\frac{1}{2}\mathbb{B}\left(  g,g\right)  }:=\mathcal{C}\left(
g\right)  .
\]
We denote $\widetilde{\mathbb{P}}$ the measure corresponding to $\mathcal{C}%
\left(  g\right)  $.

We now recall that
\[
\mathcal{C}(h)=%
{\displaystyle\int\limits_{\mathcal{L}_{\mathbb{R}}^{\prime}\left(
\mathbb{Q}_{p}\right)  }}
e^{\sqrt{-1}\left(  W,h\right)  }d\mathbb{P}\left(  W\right)  =%
{\displaystyle\int\limits_{\mathbb{R}}}
e^{\sqrt{-1}x}d\mathbb{P}_{h}\left(  x\right)  ,
\]
where $\mathbb{P}_{h}\left(  x\right)  $ is the measure of the half-space
$\left(  W,h\right)  \leq x$ in $\mathcal{L}_{\mathbb{R}}^{\prime}\left(
\mathbb{Q}_{p}\right)  $, see e.g. \cite[Chapter IV, Section 4.1]{Gelf-vol-4}.
Now if $\mathcal{C}(h_{n})\rightarrow\mathcal{C}(\widetilde{h})$, and
$\mathbb{P}_{h_{n}}\left(  \mathbb{R}\right)  \leq1$ for all $n$, then
$\mathbb{P}_{h_{n}}\Rightarrow\mathbb{P}_{\widetilde{h}}$, $\mathcal{C}%
(\widetilde{h})$ is the characteristic function of $\mathbb{P}_{\widetilde{h}%
}$, see e.g. \cite[Theorem 7.8.11]{Ash}. The arrow `$\Rightarrow$' means
\ that%
\begin{equation}%
{\displaystyle\int\limits_{\mathbb{R}}}
l(x)d\mathbb{P}_{h_{n}}\left(  x\right)  \rightarrow%
{\displaystyle\int\limits_{\mathbb{R}}}
l(x)d\mathbb{P}_{\widetilde{h}}\left(  x\right)  \text{ for any bounded
continuous function }l(x)\text{.} \label{Eq_weak_convergence}%
\end{equation}

Therefore%
\[
\widetilde{\mathbb{P}}_{L,m}\Rightarrow\widetilde{\mathbb{P}}\text{ when
}\left\vert m\right\vert _{p}\rightarrow\infty\text{, }L\rightarrow
\infty\text{.}%
\]
Now, if $l(x)\in L^{1}\left(  \mathbb{R},\mathbb{P}_{h_{n}}\right)  $ for any
$n$ and $l(x)\in L^{1}\left(  \mathbb{R},\right)  $, by using the fact that
the bounded continuous functions are dense in $L^{1}\left(  \mathbb{R}%
,\mathbb{P}_{h_{n}}\right)  $\ and $L^{1}\left(  \mathbb{R},\mathbb{P}%
_{\widetilde{h}}\right)  $, see e.g. \cite[Proposition 1.3.22]{Kondratiev et
al}, in (\ref{Eq_weak_convergence}) we can assume that $l(x)$ is an integrable function.

In conclusion, we have the following result.

\begin{lemma}
\label{Lemma_B}%
\[
\lim_{L\rightarrow\infty}\text{ }\lim_{\left\vert m\right\vert _{p}%
\rightarrow\infty}%
{\displaystyle\int\limits_{\mathcal{L}_{\mathbb{R}}\left(  \mathbb{Q}%
_{p}\right)  }}
e^{\sum_{j=1}^{N}v_{j}\left(  \delta_{I}\left(  x-x_{j}\right)  ,\widetilde
{\varphi}\right)  }d\widetilde{\mathbb{P}}_{L,m}\left(  \widetilde{\varphi
}\right)  =%
{\displaystyle\int\limits_{\mathcal{L}_{\mathbb{R}}\left(  \mathbb{Q}%
_{p}\right)  }}
e^{\sum_{j=1}^{N}v_{j}\left(  \delta_{I}\left(  x-x_{j}\right)  ,\widetilde
{\varphi}\right)  }d\widetilde{\mathbb{P}}\left(  \widetilde{\varphi}\right)
,
\]
and
\[
\lim_{I\rightarrow\infty}%
{\displaystyle\int\limits_{\mathcal{L}_{\mathbb{R}}\left(  \mathbb{Q}%
_{p}\right)  }}
e^{\sum_{j=1}^{N}v_{j}\left(  \delta_{I}\left(  x-x_{j}\right)  ,\widetilde
{\varphi}\right)  }d\widetilde{\mathbb{P}}\left(  \widetilde{\varphi}\right)
=%
{\displaystyle\int\limits_{\mathcal{L}_{\mathbb{R}}\left(  \mathbb{Q}%
_{p}\right)  }}
e^{\sum_{j=1}^{N}v_{j}\widetilde{\varphi}\left(  x_{j}\right)  }%
d\widetilde{\mathbb{P}}\left(  \widetilde{\varphi}\right)  .
\]

\end{lemma}

\subsubsection{A formula for $\mathcal{A}_{R}^{\left(  N\right)  }\left(
\boldsymbol{k}\right)  $}

Now, we recall that by using the change of variables
(\ref{eq:change of variab}), we have
\[
\widetilde{\mathcal{A}}_{R}^{\left(  N\right)  }\left(  \boldsymbol{x}%
,\boldsymbol{v};I\right)  =\frac{1}{Z_{0}^{1/D}}e^{\sum_{j=1}^{N}v_{j}\left(
\delta_{I}\left(  x-x_{j}\right)  ,-\varphi_{L,m}\right)  }%
{\displaystyle\int\limits_{\mathcal{L}_{\mathbb{R}}\left(  \mathbb{Q}%
_{p}\right)  }}
e^{\sum_{j=1}^{N}v_{j}\left(  \delta_{I}\left(  x-x_{j}\right)  ,\widetilde
{\varphi}\right)  }d\widetilde{\mathbb{P}}_{L,m}\left(  \widetilde{\varphi
}\right)  \text{,}%
\]
taking $\varphi_{L,m}$ to be the unique solution of $\boldsymbol{D}%
\varphi_{L,m}=J_{L,m}$, for each $m\in\mathbb{Q}_{p}^{\times}$. Then by
applying Lemmas \ref{Lemma3}, \ref{Lemma_A}, \ref{Lemma_B},
\begin{multline*}
\lim_{I\rightarrow\infty}\text{ }\lim_{L\rightarrow\infty}\text{ }%
\lim_{\left\vert m\right\vert _{p}\rightarrow\infty\text{,}}\widetilde
{\mathcal{A}}_{R}^{\left(  N\right)  }\left(  \boldsymbol{x},\boldsymbol{v}%
;I\right)  =\widetilde{\mathcal{A}}_{R}^{\left(  N\right)  }\left(
\boldsymbol{x},\boldsymbol{v}\right) \\
=\dfrac{1}{Z_{0}^{1/D}}\text{ }e^{\dfrac{p-1}{p\ln p}\sum_{j=1}^{N}%
\sum_{i=1,i\neq j}^{N}v_{j}v_{i}\ln\left\vert x_{j}-x_{i}\right\vert _{p}}%
{\displaystyle\int\limits_{\mathcal{L}_{\mathbb{R}}\left(  \mathbb{Q}%
_{p}\right)  }}
e^{\sum_{j=1}^{N}v_{j}\widetilde{\varphi}\left(  x_{j}\right)  }%
d\widetilde{\mathbb{P}}\left(  \widetilde{\varphi}\right)  \text{.}%
\end{multline*}

By using this formula and the definition $\mathcal{A}_{R}^{\left(  N\right)
}\left(  \boldsymbol{k}\right)  $, we establish the following result.

\begin{proposition}
\label{Prop1}The amplitude $\mathcal{A}_{R}^{\left(  N\right)  }\left(
\boldsymbol{k}\right)  $ satisfies
\[
\mathcal{A}_{R}^{\left(  N\right)  }\left(  \boldsymbol{k}\right)  =\frac
{1}{Z_{0}}%
{\displaystyle\int\limits_{B_{R}^{N}}}
\prod_{j<i}^{N}\left\vert x_{j}-x_{i}\right\vert _{p}^{2\frac{\left(
p-1\right)  }{p\ln p}\boldsymbol{k}_{i}\cdot\boldsymbol{k}_{j}}%
{\displaystyle\int\limits_{\mathcal{L}_{\mathbb{R}}^{D}\left(  \mathbb{Q}%
_{p}\right)  }}
e^{\sum_{j=1}^{N}\boldsymbol{k}_{j}\cdot\widetilde{\boldsymbol{\varphi}%
}\left(  x_{j}\right)  }d\widetilde{\mathbb{P}}_{D}\left(  \widetilde
{\boldsymbol{\varphi}}\right)  \prod_{j=0}^{N}dx_{j}\text{.}%
\]

\end{proposition}

We now introduce the `convention' that the insertion points $x_{1}$, $x_{2}%
$,\ldots{}, $x_{N-1}$, $x_{N}$, with $N\geq4$, belong to the $p$-adic
projective line, and then by using the M{\"{o}}bius group, we may take the
normalization
\[
x_{1}=0\text{, }x_{N-1}=1\text{, }x_{N}=\infty\text{.}%
\]
In our framework, the convention $x_{N}=\infty$ means that the $N$-point
amplitudes do not depend on $x_{N}$, then $\mathcal{A}_{R}^{\left(  N\right)
}\left(  \boldsymbol{k}\right)  $ takes the form
\begin{multline*}
\mathcal{A}_{R}^{\left(  N\right)  }\left(  \boldsymbol{k}\right)
=\dfrac{C_{0}}{Z_{0}}%
{\displaystyle\int\limits_{B_{R}^{N-3}}}
\prod\limits_{i=2}^{N-2}\left\vert x_{i}\right\vert _{p}^{2\dfrac{\left(
p-1\right)  }{p\ln p}\boldsymbol{k}_{1}\cdot\boldsymbol{k}_{i}}\left\vert
1-x_{i}\right\vert _{p}^{2\dfrac{\left(  p-1\right)  }{p\ln p}\boldsymbol{k}%
_{N-1}\cdot\boldsymbol{k}_{i}}\\
\times\prod_{2\leq i,j\leq N-2}\left\vert x_{j}-x_{i}\right\vert _{p}%
^{2\dfrac{\left(  p-1\right)  }{p\ln p}\boldsymbol{k}_{i}\cdot\boldsymbol{k}%
_{j}}%
{\displaystyle\int\limits_{\mathcal{L}_{\mathbb{R}}^{D}\left(  \mathbb{Q}%
_{p}\right)  }}
e^{\sum_{j=2}^{N-2}\boldsymbol{k}_{j}\cdot\widetilde{\boldsymbol{\varphi}%
}\left(  x_{j}\right)  }d\widetilde{\mathbb{P}}_{D}\left(  \widetilde
{\boldsymbol{\varphi}}\right)  \prod\limits_{j=2}^{N-2}dx_{j}\text{,}%
\end{multline*}
where the momenta vectors satisfy $\sum_{i=1}^{N}\boldsymbol{k}_{i}%
=\mathbf{0}$ and
\[
C_{0}=%
{\displaystyle\int\limits_{\mathcal{L}_{\mathbb{R}}^{D}\left(  \mathbb{Q}%
_{p}\right)  }}
e^{\boldsymbol{k}_{1}\cdot\widetilde{\boldsymbol{\varphi}}\left(  0\right)
+\boldsymbol{k}_{N-1}\cdot\widetilde{\boldsymbol{\varphi}}\left(  1\right)
}d\widetilde{\mathbb{P}}_{D}\left(  \widetilde{\boldsymbol{\varphi}}\right)
.
\]

We now consider the function
\[
\Theta(\boldsymbol{k},\boldsymbol{x}):=\Theta\left(  \boldsymbol{k}%
,x_{2},\ldots,x_{N-2}\right)  =%
{\displaystyle\int\limits_{\mathcal{L}_{\mathbb{R}}^{D}\left(  \mathbb{Q}%
_{p}\right)  }}
e^{\sum_{j=2}^{N-2}\boldsymbol{k}_{j}\cdot\widetilde{\boldsymbol{\varphi}%
}\left(  x_{j}\right)  }d\widetilde{\mathbb{P}}_{D}\left(  \widetilde
{\boldsymbol{\varphi}}\right)  \text{.}%
\]
By using that
\begin{align*}
e^{\sum_{j=2}^{N-2}\boldsymbol{k}_{j}\cdot\widetilde{\boldsymbol{\varphi}%
}\left(  x_{j}\right)  } &  =\lim_{M\rightarrow\infty}\sum\limits_{r=0}%
^{M}\frac{\left(  ^{\sum_{j=2}^{N-2}\boldsymbol{k}_{j}\cdot\widetilde
{\mathbf{\varphi}}\left(  x_{j}\right)  }\right)  ^{r}}{r!}\\
&  =\lim_{M^{\prime}\rightarrow\infty}\sum\limits_{r=0}^{M^{\prime}}%
F_{r}\left(  \boldsymbol{k},\widetilde{\boldsymbol{\varphi}}\left(
x_{2}\right)  ,\ldots,\widetilde{\boldsymbol{\varphi}}\left(  x_{N-2}\right)
\right)  ,
\end{align*}
where $F_{r}\left(  \boldsymbol{k},\widetilde{\boldsymbol{\varphi}}\left(
x_{2}\right)  ,\ldots,\widetilde{\boldsymbol{\varphi}}\left(  x_{N-2}\right)
\right)  $ is a homogeneous polynomial of degree $r$ in the variables
$k_{l,j}$, $l=0,\ldots,D-1$, $j=2,\ldots,N-2$, whose coefficients are
polynomials in the $\widetilde{\boldsymbol{\varphi}}\left(  x_{2}\right)
,\ldots,\widetilde{\boldsymbol{\varphi}}\left(  x_{N-2}\right)  $. By the
dominated convergence theorem, Corollary \ref{Cor1}, and
\[
\sum\limits_{r=0}^{M^{\prime}}\left\vert F_{r}\left(  \boldsymbol{k}%
,\widetilde{\boldsymbol{\varphi}}\left(  x_{2}\right)  ,\ldots,\widetilde
{\boldsymbol{\varphi}}\left(  x_{N-2}\right)  \right)  \right\vert \leq
e^{\sum_{j=2}^{N-2}\sum_{l=0}^{D-1}\left\vert k_{l,j}\right\vert \left\vert
\widetilde{\varphi_{l}}\left(  x_{j}\right)  \right\vert }\in L^{1}\left(
\mathcal{L}_{\mathbb{R}}^{D}\left(  \mathbb{Q}_{p}\right)  ,\widetilde
{\mathbb{P}}_{D}\right)  \text{,}%
\]
we have
\begin{align*}
\Theta(\boldsymbol{k},\boldsymbol{x}) &  =%
{\displaystyle\int\limits_{\mathcal{L}_{\mathbb{R}}^{D}\left(  \mathbb{Q}%
_{p}\right)  }}
\left\{  \lim_{M^{\prime}\rightarrow\infty}\sum\limits_{r=0}^{M^{\prime}}%
F_{r}\left(  \boldsymbol{k},\widetilde{\boldsymbol{\varphi}}\left(
x_{2}\right)  ,\ldots,\widetilde{\boldsymbol{\varphi}}\left(  x_{N-2}\right)
\right)  \right\}  d\widetilde{\mathbb{P}}_{D}\left(  \boldsymbol{\varphi
}\right)  \\
&  =\lim_{M^{\prime}\rightarrow\infty}\sum\limits_{r=0}^{M^{\prime}}%
{\displaystyle\int\limits_{\mathcal{L}_{\mathbb{R}}^{D}\left(  \mathbb{Q}%
_{p}\right)  }}
F_{r}\left(  \boldsymbol{k},\widetilde{\boldsymbol{\varphi}}\left(
x_{2}\right)  ,\ldots,\widetilde{\boldsymbol{\varphi}}\left(  x_{N-2}\right)
\right)  d\widetilde{\mathbb{P}}_{D}\left(  \widetilde{\boldsymbol{\varphi}%
}\right)  \\
&  =%
{\displaystyle\int\limits_{\mathcal{L}_{\mathbb{R}}^{D}\left(  \mathbb{Q}%
_{p}\right)  }}
\text{{}}d\widetilde{\mathbb{P}}_{D}\left(  \widetilde{\boldsymbol{\varphi}%
}\right)  \text{{}}+\sum\limits_{r=1}^{\infty}\text{ }%
{\displaystyle\int\limits_{\mathcal{L}_{\mathbb{R}}^{D}\left(  \mathbb{Q}%
_{p}\right)  }}
F_{r}\left(  \boldsymbol{k},\widetilde{\boldsymbol{\varphi}}\left(
x_{2}\right)  ,\ldots,\widetilde{\boldsymbol{\varphi}}\left(  x_{N-2}\right)
\right)  \text{{}}d\widetilde{\mathbb{P}}_{D}\left(  \widetilde
{\boldsymbol{\varphi}}\right)  \text{,}%
\end{align*}
where $\boldsymbol{x=}\left(  x_{2},\ldots,x_{N-2}\right)  $. Now by using
that $F_{r}\left(  \boldsymbol{k},\widetilde{\boldsymbol{\varphi}}\left(
x_{2}\right)  ,\ldots,\widetilde{\boldsymbol{\varphi}}\left(  x_{N-2}\right)
\right)  $ are integrable continuous functions in $\boldsymbol{x}$ for
$\boldsymbol{k}$ fixed, we conclude that
\[
G_{r}(\boldsymbol{k},\boldsymbol{x}):=%
{\displaystyle\int\limits_{\mathcal{L}_{\mathbb{R}}^{D}\left(  \mathbb{Q}%
_{p}\right)  }}
F_{r}\left(  \boldsymbol{k},\widetilde{\boldsymbol{\varphi}}\left(
x_{2}\right)  ,\ldots,\widetilde{\boldsymbol{\varphi}}\left(  x_{N-2}\right)
\right)  \text{{}}d\widetilde{\mathbb{P}}_{D}\left(  \widetilde
{\boldsymbol{\varphi}}\right)
\]
is a continuous function in $\boldsymbol{x}$. Therefore
\[
\Theta(\boldsymbol{k},\boldsymbol{x})=C+\sum\limits_{r=1}^{\infty}%
G_{r}(\boldsymbol{k},\boldsymbol{x})\text{.}%
\]
Now by using the formula given in Proposition \ref{Prop1}, and Fubini's
theorem to interchange $\int_{B_{R}^{N-3}}$ and $\sum_{r=1}^{\infty}$, we
obtain the following result.

\begin{theorem}
\label{Prop2}The amplitude $\mathcal{A}_{R}^{\left(  N\right)  }\left(
\boldsymbol{k}\right)  $ admits the following expansion in the momenta:
\begin{multline*}
\mathcal{A}_{R}^{\left(  N\right)  }\left(  \boldsymbol{k}\right)
=\dfrac{CC_{0}}{Z_{0}}%
{\displaystyle\int\limits_{B_{R}^{N-3}}}
\prod\limits_{i=2}^{N-2}\left\vert x_{i}\right\vert _{p}^{2\dfrac{(p-1)}{p\ln
p}\boldsymbol{k}_{1}\cdot\boldsymbol{k}_{j}}\left\vert 1-x_{i}\right\vert
_{p}^{2\dfrac{\left(  p-1\right)  }{\ln p}\boldsymbol{k}_{N-1}\cdot
\boldsymbol{k}_{i}}\\
\times\prod_{2\leq i,j\leq N-2}\left\vert x_{j}-x_{i}\right\vert _{p}%
^{2\dfrac{\left(  p-1\right)  }{p\ln p}\boldsymbol{k}_{i}\cdot\boldsymbol{k}%
_{j}}\prod_{j=2}^{N-2}dx_{j}\\
+\dfrac{C_{0}}{Z_{0}}\sum\limits_{r=1}^{\infty}%
{\displaystyle\int\limits_{B_{R}^{N-3}}}
\prod\limits_{i=2}^{N-2}\left\vert x_{i}\right\vert _{p}^{2\dfrac{\left(
p-1\right)  }{p\ln p}\boldsymbol{k}_{1}\cdot\boldsymbol{k}_{j}}\left\vert
1-x_{i}\right\vert _{p}^{2\dfrac{\left(  p-1\right)  }{p\ln p}\boldsymbol{k}%
_{N-1}\cdot\boldsymbol{k}_{i}}\\
\times\prod_{2\leq i,j\leq N-2}\left\vert x_{j}-x_{i}\right\vert _{p}%
^{2\dfrac{\left(  p-1\right)  }{p\ln p}\boldsymbol{k}_{i}\cdot\boldsymbol{k}%
_{j}}G_{r}(\boldsymbol{k},\boldsymbol{x})\prod_{j=2}^{N-2}dx_{j}\text{.}%
\end{multline*}

\end{theorem}

To continue the study of the amplitudes $\mathcal{A}_{R}^{\left(  N\right)
}\left(  \boldsymbol{k}\right)  $, we introduce the following notation:
\begin{multline*}
A_{R}^{\left(  N\right)  }\left(  \boldsymbol{k}\right)  =\dfrac{CC_{0}}%
{Z_{0}}\int\nolimits_{B_{R}^{N-3}}\prod\limits_{i=2}^{N-2}\left\vert
x_{i}\right\vert _{p}^{2\dfrac{(p-1)}{\ln p}\boldsymbol{k}_{1}\cdot
\boldsymbol{k}_{j}}\left\vert 1-x_{i}\right\vert _{p}^{2\dfrac{\left(
p-1\right)  }{\ln p}\boldsymbol{k}_{N-1}\cdot\boldsymbol{k}_{i}}\\
\times\prod_{2\leq i,j\leq N-2}\left\vert x_{j}-x_{i}\right\vert _{p}%
^{2\dfrac{\left(  p-1\right)  }{\ln p}\boldsymbol{k}_{i}\cdot\boldsymbol{k}%
_{j}}\prod_{j=2}^{N-2}dx_{j}\text{,}%
\end{multline*}%
\begin{multline*}
Z_{G_{r},R}^{(N)}(\boldsymbol{k})=\dfrac{C_{0}}{Z_{0}}%
{\displaystyle\int\limits_{\mathbb{Q}_{p}^{N-3}}}
\prod\limits_{i=2}^{N-2}\left\vert x_{i}\right\vert _{p}^{2\dfrac{(p-1)}{p\ln
p}\boldsymbol{k}_{1}\cdot\boldsymbol{k}_{j}}\left\vert 1-x_{i}\right\vert
_{p}^{2\dfrac{\left(  p-1\right)  }{p\ln p}\boldsymbol{k}_{N-1}\cdot
\boldsymbol{k}_{i}}\\
\times\prod_{2\leq i,j\leq N-2}\left\vert x_{j}-x_{i}\right\vert _{p}%
^{2\dfrac{\left(  p-1\right)  }{p\ln p}\boldsymbol{k}_{i}\cdot\boldsymbol{k}%
_{j}}1_{B_{R}^{N-3}}(\boldsymbol{x})G_{r}(\boldsymbol{k},\boldsymbol{x}%
)\prod_{j=2}^{N-2}dx_{j}\text{.}%
\end{multline*}
Notice that $1_{B_{R}^{N-3}}(\boldsymbol{x})G_{r}(\boldsymbol{k}%
,\boldsymbol{x})$ is a continuous function in $\boldsymbol{x}$ with support
contained in $B_{R}^{N-3}$.

\section{Regularization of $p$-adic open string amplitudes, and multivariate
local zeta functions}

\subsection{The $p$-adic Koba-Nielsen local zeta functions}

Take $N\geq4$ and $s_{ij}\in\mathbb{C}$ satisfying $s_{ij}=s_{ji}$ for $1\leq
i<j\leq N-1$. The $p$-adic Koba-Nielsen local zeta function (or $p$-adic open
string $N$-point zeta function) is defined as
\begin{equation}
Z^{(N)}\left(  \mathbf{s}\right)  =%
{\displaystyle\int\limits_{\mathbb{Q}_{p}^{N-3}\smallsetminus\Lambda}}
{\prod_{i=2}^{N-2}}\left\vert x_{i}\right\vert _{p}^{s_{1i}}\left\vert
1-x_{i}\right\vert _{p}^{s_{(N-1)i}}\text{{}}{\prod_{2\leq i<j\leq N-2}%
}\left\vert x_{i}-x_{j}\right\vert _{p}^{s_{ij}}{\prod_{i=2}^{N-2}}dx_{i},
\label{zeta_funtion_string}%
\end{equation}
where $\boldsymbol{s}=\left(  s_{ij}\right)  \in\mathbb{C}^{D_{0}}$, here
$D_{0}$ denotes the total number of possible subsets $\left\{  i,j\right\}  $,
$%
{\textstyle\prod_{i=2}^{N-2}}
dx_{i}$ is the normalized Haar measure of $\mathbb{Q}_{p}^{N-3}$, and
\[
\Lambda:=\left\{  \left(  x_{2},\ldots,x_{N-2}\right)  \in\mathbb{Q}_{p}%
^{N-3};\text{{}}{\prod_{i=2}^{N-2}}x_{i}\left(  1-x_{i}\right)  \text{{}%
}{\prod_{2\leq i<j\leq N-2}}\left(  x_{i}-x_{j}\right)  =0\right\}  \text{.}%
\]
These functions were introduced in \cite{Zun-B-C-LMP}, see also
\cite{Bocardo:2020mk}. The functions $Z^{(N)}\left(  \mathbf{s}\right)  $ are
holomorphic in a certain domain of $\mathbb{C}^{D_{0}}$ and admit analytic
continuations to $\mathbb{C}^{D_{0}}$ (denoted also as $Z^{\left(  N\right)
}\left(  \mathbf{s}\right)  $) as rational functions in the variables
\[
p^{-s_{ij}},i,j\in\left\{  1,\ldots,N-1\right\}  ,
\]
see \cite[Theorem 1]{Zun-B-C-LMP}, \cite[Theorem 6.1]{Bocardo:2020mk}.

If $\phi\left(  x_{2},\ldots,x_{N-2}\right)  $ is a locally constant function
with compact support, then
\begin{gather*}
Z_{\phi}^{(N)}(\boldsymbol{s})=\\%
{\displaystyle\int\limits_{\mathbb{Q}_{p}^{N-3}\smallsetminus\Lambda}}
\phi\left(  x_{2},\ldots,x_{N-2}\right)  {\prod_{i=2}^{N-2}}\left\vert
x_{i}\right\vert _{p}^{s_{1i}}\left\vert 1-x_{i}\right\vert _{p}^{s_{\left(
N-1\right)  i}}{\prod_{2\leq i<j\leq N-2}}\left\vert x_{i}-x_{j}\right\vert
_{p}^{s_{i}\text{{}}_{j}}{\prod_{i=2}^{N-2}}dx_{i}\text{,}%
\end{gather*}
for $\operatorname*{Re}(s_{ij})>0$ for any $ij$, is a multivariate Igusa local
zeta function. These functions admit analytic continuations as rational
functions of the variables $p^{-s_{ij}}$, \cite{Loeser}. If we take $\phi$ to
be the characteristic function of\ $B_{R}^{N-3}$, the ball centered at the
origin with radius $p^{R}$, the dominated convergence theorem and
\cite[Theorem 1]{Zun-B-C-LMP}, imply that
\begin{gather}
\lim_{R\rightarrow\infty}Z_{R}^{(N)}(\boldsymbol{s}):=\lim_{R\rightarrow
\infty}%
{\displaystyle\int\limits_{B_{R}^{N-3}\smallsetminus\Lambda}}
{\prod_{i=2}^{N-2}}\left\vert x_{i}\right\vert _{p}^{s_{1i}}\left\vert
1-x_{i}\right\vert _{p}^{s_{\left(  N-1\right)  i}}{\prod_{2\leq i<j\leq N-2}%
}\left\vert x_{i}-x_{j}\right\vert _{p}^{s_{i}\text{{}}_{j}}{\prod_{i=2}%
^{N-2}}dx_{i}\label{REsult_A}\\
=Z^{(N)}\left(  \mathbf{s}\right)  ,\nonumber
\end{gather}
for any $\boldsymbol{s}$ in the natural domain of $Z^{(N)}\left(
\mathbf{s}\right)  $.

In \cite{B-F-O-W}, Brekke, Freund, Olson and Witten work out the $N$-point
amplitudes in explicit form and investigate how these can be obtained from an
effective Lagrangian. The $p$-adic open string $N$-point tree amplitudes are
defined as
\begin{gather}
A_{\mathcal{M}}^{(N)}\left(  \boldsymbol{k}\right)  =\label{Amplitude}\\%
{\displaystyle\int\limits_{\mathbb{Q}_{p}^{N-3}}}
{\displaystyle\prod\limits_{i=2}^{N-2}}
\left\vert x_{i}\right\vert _{p}^{\boldsymbol{k}_{1}\boldsymbol{k}_{i}%
}\left\vert 1-x_{i}\right\vert _{p}^{\boldsymbol{k}_{N-1}\boldsymbol{k}_{i}%
}\text{ }%
{\displaystyle\prod\limits_{2\leq i<j\leq N-2}}
\left\vert x_{i}-x_{j}\right\vert _{p}^{\boldsymbol{k}_{i}\boldsymbol{k}_{j}}%
{\displaystyle\prod\limits_{i=2}^{N-2}}
dx_{i}\text{,}\nonumber
\end{gather}
where $%
{\textstyle\prod\nolimits_{i=2}^{N-2}}
dx_{i}$ is the normalized Haar measure of $\mathbb{Q}_{p}^{N-3}$,
$\boldsymbol{k}=\left(  \boldsymbol{k}_{1},\ldots,\boldsymbol{k}_{N}\right)
$, $\boldsymbol{k}_{i}=\left(  k_{0,i},\ldots,k_{D-1,i}\right)  $,
$i=1,\ldots,N$, $N\geq4$, is the momentum vector of the $i$-th tachyon (with
Minkowski product $\boldsymbol{k}_{i}\boldsymbol{k}_{j}=-k_{0,i}%
k_{0,j}+k_{1,i}k_{1,j}+\cdots+k_{D-1,i}k_{D-1,j}$) obeying
\[
\sum_{i=1}^{N}\boldsymbol{k}_{i}=\boldsymbol{0}\text{, \ \ \ \ \ }%
\boldsymbol{k}_{i}\boldsymbol{k}_{i}=2\text{ \ for }i=1,\ldots,N.
\]
In \cite{Zun-B-C-LMP}, \cite{Bocardo:2020mk}, the $p$-adic open string
$N$-point tree integrals $Z^{(N)}(\boldsymbol{s})$ are used as regularizations
of the amplitudes $A_{\mathcal{M}}^{(N)}\left(  \boldsymbol{k}\right)  $. More
precisely, the amplitude $A_{\mathcal{M}}^{(N)}\left(  \boldsymbol{k}\right)
$ can be re-define as
\[
A_{\mathcal{M}}^{(N)}\left(  \boldsymbol{k}\right)  =Z^{(N)}(\boldsymbol{s}%
)\mid_{s_{ij}=\boldsymbol{k}_{i}\boldsymbol{k}_{j}}\text{with }i\in\left\{
1,\ldots,N-1\right\}  \text{, }j\in T\text{ or }i,j\in T,
\]
where $T=\left\{  2,\ldots,N-2\right\}  $. Then the amplitudes $A_{\mathcal{M}%
}^{(N)}\left(  \boldsymbol{k}\right)  $ are well-defined rational functions of
the variables $p^{-\boldsymbol{k}_{i}\boldsymbol{k}_{j}}$, $i$, $j\in\left\{
1,\ldots,N-1\right\}  $, which agree with integrals (\ref{Amplitude}) when
they converge.

\begin{remark}
In \cite{Zun-B-C-LMP}, \cite{Bocardo:2020mk}, the local zeta functions
$Z^{(N)}(\boldsymbol{s})$ were used to regularize Koba-Nielsen amplitudes
$A_{\mathcal{M}}^{(N)}\left(  \boldsymbol{k}\right)  $, when the momenta
$\boldsymbol{k}$ belong to the Minkowski space. In this article, we use \ the
functions $Z^{(N)}(\boldsymbol{s})$ to regularize Koba-Nielsen amplitudes
$A^{(N)}\left(  \boldsymbol{k}\right)  $ when the momenta $\boldsymbol{k}$
belong to the Euclidean space. This is possible because $Z^{(N)}%
(\boldsymbol{s})$ is a rational function in the variables $p^{-s_{ij}}$,
$s_{ij}\in\mathbb{C}$, for $i$, $j\in\left\{  1,\ldots,N-1\right\}  $.
\end{remark}

\begin{remark}
\label{Nota_poles}We denote by $Z_{\cdot}^{(N)}(\boldsymbol{s})$ the
distribution $\phi\rightarrow Z_{\phi}^{(N)}(\boldsymbol{s})$. Then the
mapping
\begin{equation}%
\begin{array}
[c]{ccc}%
\mathbb{C}^{D_{0}} & \rightarrow & \mathcal{D}^{\prime}(\mathbb{Q}_{p}%
^{N-3})\\
&  & \\
\boldsymbol{s} & \rightarrow & Z_{\cdot}^{(N)}(\boldsymbol{s})
\end{array}
\label{Map_1}%
\end{equation}
is a meromorphic function of $\boldsymbol{s}$. By using the fact that
$\mathcal{D}(\mathbb{Q}_{p}^{N-3})$ is dense in the space of continuous
functions with compact support $\mathcal{C}_{c}(\mathbb{Q}_{p}^{N-3})$, the
functional $\phi\rightarrow Z_{\phi}^{(N)}(\boldsymbol{s})$ has a unique
extension to $\mathcal{C}_{c}(\mathbb{Q}_{p}^{N-3})$. Furthermore,if
$\boldsymbol{s}_{0}$ is a pole of $Z_{\phi}^{(N)}(\boldsymbol{s})$, by using
Gel'fand-Shilov method of analytic continuation, see e.g. \cite[pgs.
65-67]{Igusa},%
\[
Z_{\phi}^{(N)}(\boldsymbol{s})=%
{\displaystyle\sum\limits_{\boldsymbol{k}\in\mathbb{Z}^{D_{0}}}}
c_{\boldsymbol{k}}\left(  \phi\right)  \left(  \boldsymbol{s-s}_{0}\right)
^{\boldsymbol{k}},
\]
where the $c_{\boldsymbol{k}}$s are distributions from $\mathcal{D}^{\prime
}(\mathbb{Q}_{p}^{N-3})$. The density of $\mathcal{D}(\mathbb{Q}_{p}^{N-3})$
in $\mathcal{C}_{c}(\mathbb{Q}_{p}^{N-3})$ implies that $c_{\boldsymbol{k}%
}\neq0$ in $\mathcal{D}^{\prime}(\mathbb{Q}_{p}^{N-3})$ if and only if
$c_{\boldsymbol{k}}\neq0$ in $\mathcal{C}_{c}^{\prime}(\mathbb{Q}_{p}^{N-3})$,
the strong dual space of $\mathcal{C}_{c}(\mathbb{Q}_{p}^{N-3})$. This implies
that the mapping%
\[%
\begin{array}
[c]{ccc}%
\mathbb{C}^{D_{0}} & \rightarrow & \mathcal{C}_{c}^{\prime}(\mathbb{Q}%
_{p}^{N-3})\\
&  & \\
\boldsymbol{s} & \rightarrow & Z_{\cdot}^{(N)}(\boldsymbol{s})
\end{array}
\]
is a meromorphic function in $\boldsymbol{s}$ having the same poles of the
mapping (\ref{Map_1}).
\end{remark}

\subsection{The limit $\lim_{R\rightarrow\infty}\mathcal{A}_{R}^{\left(
N\right)  }\left(  \boldsymbol{k}\right)  $}

We now apply the above-mentioned results to study the limit
\[
\lim_{R\rightarrow\infty}\mathcal{A}_{R}^{\left(  N\right)  }\left(
\boldsymbol{k}\right)  \text{.}%
\]
First, notice that by (\ref{REsult_A}),
\[
\frac{Z_{0}}{CC_{0}}\lim_{R\rightarrow\infty}A_{R}^{\left(  N\right)  }\left(
\boldsymbol{k}\right)  =\lim_{R\rightarrow\infty}\left(  Z_{R}^{(N)}%
(\boldsymbol{s})\mid_{s_{ij}=2\frac{(p-1)}{p\ln p}\boldsymbol{k}_{i}%
\cdot\boldsymbol{k}_{j}}\right)  =Z^{(N)}\left(  \boldsymbol{s}\right)
\mid_{s_{ij}=2\frac{\left(  p-1\right)  }{p\ln p}\boldsymbol{k}_{i}%
\cdot\boldsymbol{k}_{j}}\text{.}%
\]
Now by using the fact that $Z^{(N)}\left(  \boldsymbol{s}\right)  $ is a
holomorphic function in a certain domain of $\mathbb{C}^{D_{0}}$, we conclude
that $\lim_{R\rightarrow\infty}A_{R}^{\left(  N\right)  }\left(
\boldsymbol{k}\right)  $ exists for $\boldsymbol{k}$ belonging a non-empty
subset of $\mathbb{C}^{D_{0}}$.

Second, by using Remark \ref{Nota_poles}, we may assume that $1_{B_{R}^{N-3}%
}(\boldsymbol{x})G_{r}(\boldsymbol{k},\boldsymbol{x})=\phi$ is a test function
in $\boldsymbol{x}$,\ and then $Z_{G_{r},R}^{(N)}(\boldsymbol{k})=\frac{C_{0}%
}{Z_{0}}Z_{\phi}^{(N)}(\boldsymbol{s})\mid_{s_{ij}=2\frac{\left(  p-1\right)
}{p\ln p}\boldsymbol{k}_{i}\cdot\boldsymbol{k}_{j}}$ is a multivariate local
zeta function. Furthermore,
\begin{multline*}
|Z_{G_{r},R}^{(N)}(\boldsymbol{k})|\leq\frac{C_{0}}{Z_{0}}%
{\displaystyle\int\limits_{\mathbb{Q}_{p}^{N-3}}}
\prod\limits_{i=2}^{N-2}\left\vert x_{i}\right\vert _{p}^{2\dfrac{\left(
p-1\right)  }{p\ln p}\boldsymbol{k}_{1}\cdot\boldsymbol{k}_{j}}\left\vert
1-x_{i}\right\vert _{p}^{2\dfrac{\left(  p-1\right)  }{p\ln p}\boldsymbol{k}%
_{N-1}\cdot\boldsymbol{k}_{i}}\\
\times\prod_{2\leq i,j\leq N-2}\left\vert x_{j}-x_{i}\right\vert _{p}%
^{2\dfrac{\left(  p-1\right)  }{p\ln p}\boldsymbol{k}_{i}\cdot\boldsymbol{k}%
_{j}}\left\vert G_{r}(\boldsymbol{k},\boldsymbol{x})\right\vert \prod
_{j=2}^{N-2}dx_{j}\text{,}%
\end{multline*}
which implies that $\left\vert Z_{G_{r},R}^{(N)}(\boldsymbol{k})\right\vert
\leq\frac{C_{0}\text{ }C_{r}\left(  \boldsymbol{k},R\right)  }{Z_{0}}%
Z^{(N)}\left(  \boldsymbol{k}\right)  $, where
\[
C_{r}\left(  \boldsymbol{k},R\right)  =\sup_{\boldsymbol{x}\in B_{R}^{N-3}%
}\left\vert G_{r}(\boldsymbol{k},\boldsymbol{x})\right\vert \text{.}%
\]
Since $Z^{(N)}\left(  \boldsymbol{k}\right)  $\ converges in a non-empty open
set, we conclude that all the $Z_{G_{r},R}^{(N)}(\boldsymbol{k})$s converges
in the open set where $Z^{(N)}\left(  \boldsymbol{k}\right)  $\ converges.

In conclusion we have the following result.

\begin{theorem}
\label{Prop3}The amplitudes $\mathcal{A}_{R}^{\left(  N\right)  }\left(
\boldsymbol{k}\right)  $\ satisfy the following. For $R$ fixed,
\[
\mathcal{A}_{R}^{\left(  N\right)  }\left(  \boldsymbol{k}\right)
=A_{R}^{\left(  N\right)  }\left(  \boldsymbol{k}\right)  +\sum\limits_{r=1}%
^{\infty}Z_{G_{r},R}^{(N)}(\boldsymbol{k})\text{,}%
\]
where $A_{R}^{\left(  N\right)  }\left(  \boldsymbol{k}\right)  $, and all the
$Z_{G_{r},R}^{(N)}(\boldsymbol{k})$s are multivariate Igusa's local zeta
functions,\ all of them converging in a common non-empty open set.
Furthermore,
\[
\lim_{R\rightarrow\infty}A_{R}^{\left(  N\right)  }\left(  \boldsymbol{k}%
\right)  =\frac{CC_{0}}{Z_{0}}Z^{(N)}\left(  \boldsymbol{k}\right)  \text{,}%
\]
which is the $p$-adic Koba-Nielsen open string amplitude.
\end{theorem}

\subsection{$\boldsymbol{\varphi}^{4}$-theories}

Consider the family of $\boldsymbol{\varphi}^{4}$-interacting quantum field
theories:
\[
\frac{1_{\mathcal{L}_{\mathbb{R}}^{D}\left(  \mathbb{Q}_{p}\right)  }\left(
\boldsymbol{\varphi}\right)  \text{{}}e^{-\lambda E_{int}(\boldsymbol{\varphi
})}d\mathbb{P}_{D}\left(  \boldsymbol{\varphi}\right)  }{Z}\text{, for
}\lambda>0\text{,}%
\]
where
\[
E_{int}(\boldsymbol{\varphi})=\sum\limits_{j=0}^{D-1}%
{\displaystyle\int\limits_{\mathbb{Q}_{p}}}
\varphi_{j}^{4}(x)dx\text{,\ and\ \ }Z=%
{\displaystyle\int\limits_{\mathcal{L}_{\mathbb{R}}^{D}\left(  \mathbb{Q}%
_{p}\right)  }}
e^{-\lambda E_{int}(\boldsymbol{\varphi})}d\mathbb{P}_{D}\left(
\boldsymbol{\varphi}\right)  \text{.}%
\]
The amplitudes of such theories are defined as
\[
\mathcal{A}_{R}^{\left(  N\right)  }\left(  \boldsymbol{k},\lambda\right)
=\frac{1}{Z}%
{\displaystyle\int\limits_{B_{R}^{N-3}}}
\left\{
{\displaystyle\int\limits_{\mathcal{L}_{\mathbb{R}}^{D}\left(  \mathbb{Q}%
_{p}\right)  }}
e^{\sum_{j=2}^{N-2}\boldsymbol{k}_{j}\cdot\boldsymbol{\varphi}\left(
x_{j}\right)  -\lambda E_{int}(\boldsymbol{\varphi})}d\mathbb{P}_{D}\left(
\boldsymbol{\varphi}\right)  \right\}  \prod_{j=2}^{N-2}dx_{j}\text{.}%
\]
These amplitudes admit expansions of the type given in Proposition
\ref{Prop2}, where the functions $G_{r}(\boldsymbol{k},\boldsymbol{x})$ are
replaced by continuous functions in $\boldsymbol{x}$ depending on
$\boldsymbol{k}$ and $\lambda$. The behavior of these quantum field theories
is completely different from the standard ones due to the fact that we are
computing the correlation functions for a very particular class of
observables, which are products of vertex operators.

\textbf{Conflicts of Interest:} The authors declare no conflict of interest.

\textbf{ Author Contributions:} all the authors contributed to the manuscript
equally. All authors have read and agreed to the published version of the manuscript.

\begin{acknowledgement}
The authors wish to thank the referee for his/her careful reading of the
original manuscript and the suggestions made.
\end{acknowledgement}


\begin{thebibliography}{99}                                                                                               %
\bigskip

\bibitem {Veneziano}G. Veneziano, Construction of a crossing-symmetric,
Reggeon-behaved amplitude for linearly rising trajectories, Il Nuovo Cimento
A. 57 (1968) 190-197.

\bibitem {Virasoro}M. Virasoro, Alternative constructions of
crossing-symmetric amplitudes with Regge behavior, Phys. Rev. \ 177 (5) (1969)\ 2309--2311.

\bibitem {Koba-Nielsen}Z. Koba, H. Nielsen, Reaction amplitude for N-mesons: A
generalization of the Veneziano-Bardak{\c c}i-Ruegg-Virasoro model, Nuclear
Physics B. \ 10 (4) (1969) 633--655.

\bibitem {B-F-O-W}Lee Brekke, Peter G. O. Freund, Mark Olson, Edward Witten,
Non-Archimedean string dynamics, Nuclear Phys. B 302 (3) (1988) 365--402.

\bibitem {Vol}I. V. Volovich, $p$-Adic string, Classical Quantum Gravity 4 (4)
(1987) L83--L87.

\bibitem {Bocardo:2020mk}M. Bocardo-Gaspar, Willem Veys, W. A. Zu{\~n}%
iga-Galindo, Meromorphic continuation of Koba-Nielsen string amplitudes, \ J.
High Energy Phys. 138 (2020).

\bibitem {Compean:2020B-field}H. Garc{\'\i}a-Compe{\'a}n, Edgar Y. L{\'o}pez,
W. A. Z{\'u}{\~n}iga-Galindo, $p$-Adic open string amplitudes with Chan-Paton
factors coupled to a constant B-field, Nucl.\ Phys.\ B 951 (2020) 114904.

\bibitem {Zun-B-C-LMP}M. Bocardo-Gaspar, H. Garcia-Compean, W. A. Z{\'u}%
{\~n}iga-Galindo, Regularization of $p$-adic String Amplitudes, and
Multivariate Local Zeta Functions, Lett. Math. Phys. 109 (5) (2019) 1167--1204.

\bibitem {Zun-B-C-JHEP}M. Bocardo-Gaspar, H. Garcia-Compean, W. A. Z{\'u}%
{\~n}iga-Galindo, $p$-Adic string amplitudes in the limit $p$ approaches to
one, J. High Energy Phys. 43 (2018).

\bibitem {Symmetry}M. Bocardo-Gaspar, H. Garc{\'{\i}}a-Compe{\'{a}}n, E.Y.
L{\'{o}}pez, W.A. Z{\'{u}}{\~{n}}iga-Galindo, Local Zeta Functions and
Koba--Nielsen String Amplitudes, Symmetry 13 (2021) 967. https://doi.org/10.3390/sym13060967.

\bibitem {Arefeva-1}I. Ya. Aref'eva, B. G. Dragovi\'{c}, I. V. Volovich, On
the adelic string amplitudes, Phys. Lett. B (209) (4) (1988) 445--450.

\bibitem {Arefeva-2}I. Ya. Aref'eva, B. G. Dragovi\'{c}, I. V. Volovich, Open
and closed $p$-adic strings and quadratic extensions of number fields, Phys.
Lett. B (212) (3) (1988) 283--291.

\bibitem {Arefeva-3}I. Ya. Aref'eva, B. G. Dragovi\'{c}, I. V. Volovich, $p$
-adic superstrings, Phys. Lett. B (214) (3) (1988) 339--349.

\bibitem {Brekke et al}Lee Brekke, Peter G. O. Freund, $p$-Adic numbers in
physics, Phys. Rep. 233 (1) (1993) 1--66.

\bibitem {Hloseuk-Spect}Zvonimir Hlou\v{s}ek, Donald Spector, $p$-Adic string
theories, Ann. Physics 189 (2) (1989) 370--431.

\bibitem {V-V-Z}V. S. Vladimirov, I. V. Volovich, E. I. Zelenov, $p$-adic
analysis and mathematical physics, World Scientific, Singapore, 1994.

\bibitem {Dragovich:2017kge}B. Dragovich, A. Y. Khrennikov, S. V. Kozyrev, I.
V. Volovich, E. I. Zelenov, $p$-Adic mathematical physics: the first 30 years,
$p$-Adic Num. Ultrametr. Anal. Appl. 9 (2017) 87--121.

\bibitem {Gubser:2016guj}S. S. Gubser, J. Knaute, S. Parikh, A. Samberg, and
P. Witaszczyk, $p$-Adic AdS/CFT, \ Commun.\ Math.\ Phys.\ 352 (3) (2017) 1019.

\bibitem {Heydeman:2016ldy}M. Heydeman, M. Marcolli, I. Saberi, and B. Stoica,
Tensor networks, $p$-adic fields, and algebraic curves: arithmetic and the
AdS$\mbox{}_{3}$/CFT$\mbox{}_{2}$ correspondence,\ Adv. Theor. Math. Phys.22
(1) (2018) 93--176.

\bibitem {Gubser:2016htz}S. S. Gubser, M.\ Heydeman, C. Jepsen, M. Marcolli,
S. Parikh, I. Saberi, B. Stoica, and B. Trundy, Edge length dynamics on graphs
with applications to $p$-adic AdS/CFT, \ J. High Energy Phys.157 (2017).

\bibitem {Dutta:2017bja}P. Dutta, D. Ghoshal, and A. Lala, Notes on the
Exchange Interactions in Holographic $p$-adic CFT, \ Physics Letters B773
(2017) 283-289.

\bibitem {Freund:1987kt}P. G. O. Freund, M. Olson,\ Nonarchimedean Strings,
Phys.\ Lett.\ B 199 (2) (1987) 186-190.

\bibitem {Freund:1987ck}P. G. O. Freund, E. Witten, Adelic String Amplitudes,
Phys.\ Lett.\ B 199 (2) (1987) 191-194.

\bibitem {Gerasimov:2000zp}A. A. Gerasimov, S. L. Shatashvili, On exact
tachyon potential in open string field theory, J. High Energ. Phys. 5 (2013).

\bibitem {Spokoiny:1988zk}B. L. Spokoiny, Quantum Geometry of Nonarchimedean
Particles and Strings, Phys.\ Lett.\ B 208 (3-4) (1988)\textbf{ }401-405.

\bibitem {Ghoshal:2006}D. Ghoshal, $p$-Adic string theories provide lattice
discretization to the ordinary string worldsheet, Phys. Rev. Lett. 97 (2006)
151601.{}

\bibitem {Brekke:1988dg}L. Brekke, P. G. O. Freund, M. \ Olson, E. Witten,
\ Non-archimedean String Dynamics, Nucl.\ Phys.\ B 302 (3) (1988) 365-402.

\bibitem {Frampton:1987sp}P. H. Frampton, Y. Okada, The $p$-adic String
$N$-Point Function,\textit{\ }Phys.\ Rev.\ Lett.\ 60, (1988) 484-486.

\bibitem {Zhang}R. B. Zhang, Lagrangian formulation of open and closed
$p$-adic strings, Phys. Lett. B 209, (2-3) (1988) 229-232.

\bibitem {Parisi}G. Parisi, On $p$-adic functional integrals, Mod. Phys. Lett.
A3 (1988) 639-643.

\bibitem {Zabrodin:1988ep}A. V. Zabrodin, Nonarchimedean Strings and
Bruhat-tits Trees, Commun.\ Math.\ Phys.\textit{ }\ 123 (1989) 463--483.

\bibitem {Ghoshal:2004ay}D. Ghoshal, T. Kawano, Towards $p$-adic string in
constant $B$-field, Nucl.\ Phys.\ B \ 710 (2005) 577-598.

\bibitem {Jaffe-Glimm}J. Glimm, A. Jaffe, Quantum Physics. A Functional
Integral Point of View, Second edition, Springer, New York 1987.

\bibitem {Arroyo-Zuniga}Edilberto Arroyo-Ortiz, W. A. Z{\'{u}}{\~{n}}%
iga-Galindo, Construction of $p$-adic Covariant Quantum Fields in the
Framework of White Noise Analysis, Rep. Math. Phys. 84 (1) (2019) 1--34.

\bibitem {Zuniga-JFAA}W. A. Z{\'{u}}{\~{n}}iga-Galindo, Non-Archimedean white
noise, pseudodifferential stochastic equations, and massive Euclidean fields,
J. Fourier Anal. Appl. 23 (2) (2017) 288--323.

\bibitem {Zuniga-Preprint}W. A. Z{\'{u}}{\~{n}}iga-Galindo, Non-Archimedean
statistical field theory. arXiv:2006.05559.

\bibitem {Igusa}J.-I Igusa, An introduction to the theory of local zeta
functions; AMS/IP Studies in Advanced Mathematics, 14. American Mathematical
Society, Providence, RI; International Press, Cambridge, MA, 2000.

\bibitem {Abouelsaood:1986gd}A. Abouelsaood, C. G. Callan, Jr., C. R. Nappi
and S. A. Yost, Open Strings in Background Gauge Fields,\ Nucl. Phys. B 280
(1987) 599-624. doi:10.1016/0550-3213(87)90164-7

\bibitem {Alberio et al}S. Albeverio, A. Yu. Khrennikov, V. M. Shelkovich,
Theory of $p$-adic distributions: linear and nonlinear models. London
Mathematical Society Lecture Note Serie s, 370; Cambridge University Press:
Cambridge, 2010.

\bibitem {Taibleson}M. H. Taibleson, Fourier analysis on local fields,
Princeton University Press, Princeton, N.J., 1975.

\bibitem {Bruhat}Fran{\c{c}}ois Bruhat, Distributions sur un groupe localement
compact et applications {\`{a}} l'{\'{e}}tude des repr{\'{e}}sentations des
groupes p-adiques, Bull. Soc. Math. France 89 (1961) 43--75.

\bibitem {hidawhitenoise}Hida Takeyuki, Kuo Hui-Hsiung, Potthoff J{\"{u}}rgen,
Streit Ludwig,White noise. An infinite dimensional calculus, Dordrecht, Kluwer
Academic Publishers, 1993.

\bibitem {Gelf-vol-4}I. M. Gel 'fand, N. Ya. Vilenkin, Generalized functions.
Vol. 4. Applications of harmonic analysis. New York - London, academic Press (1964).

\bibitem {Ash}Robert B. Ash, Probability and measure theory. Second edition.
With contributions by Catherine Dol\'{e}ans-Dade. Harcourt/Academic Press,
Burlington, MA, 2000.

\bibitem {Kondratiev et al}Sergio Albeverio, Yuri Kondratiev, Yuri Kozitsky,
Michael R\"{o}ckner, The statistical mechanics of quantum lattice systems. A
path integral approach. EMS Tracts in Mathematics, 8. European Mathematical
Society (EMS), Z\"{u}rich, 2009.

\bibitem {Loeser}F. \ Loeser, Fonctions z\^{e}ta locales d'Igusa \`{a}
plusieurs variables, int\'{e}gration dans les fibres, et discriminants,
\textit{A}nn. Sci. \'{E}cole Norm. Sup. (22) (3) (1989) 435--471.
\end{thebibliography}
\end{document}